\documentclass[final,leqno]{siamltex}

\usepackage{amsmath,amsfonts,graphicx}
\pdfoutput=1

\usepackage{verbatim}

\def\tbe{ {\tilde\beta}}
\def\ls{ l_{123}}
\def\ms{ m_{123}}
\def\ns{ n_{123}}
\def\l2{ l_{23}}
\def\m2{ m_{23}}
\def\n2{ n_{23}}

\def\fp{ f_+}
\def\fm{ f_-}
\def\De{ {D_E}}
\def\Dt{ {D_T}}
\def\Da{ {D_A}}

\def\tg{\mbox{\boldmath{$\tau$}}}
\def\norm{ {\hat{\bf n}}}
\def\tang{ {\hat\tg}}
\def\dt{ \frac{d}{dt}}
\def\be{ \mbox{\ss}}

\title{Anisotropic step stiffness from a kinetic model of epitaxial growth}

\author{Dionisios Margetis\thanks{Department of Mathematics and Institute for Physical Science and Technology, 
         University of Maryland, College Park, MD 20742-4015 ({\tt dio@math.umd.edu}).}\and
        Russel E. Caflisch\thanks{Department of Mathematics and Department
        of Materials Science \& Engineering, University of California, 405 Hilgard Avenue,
        Los Angeles, CA 90095-1555 ({\tt caflisch@math.ucla.edu}). This work was supported in part by the
         National Science Foundation through grant DMS-0402276.}}        
\begin{document}

\maketitle

\begin{abstract}
Starting from a detailed  model for the kinetics of a step edge or island boundary, 
we derive a Gibbs-Thomson type formula and the associated step stiffness as a function of the step edge orientation angle, 
$\theta$. Basic ingredients of the model are:
(i) the diffusion of point defects (``adatoms'') on terraces and along step edges; 
(ii) the convection of kinks along step edges; and 
(iii) constitutive laws that relate adatom fluxes, sources for kinks, and the kink velocity 
with densities via a mean-field approach. 
This model has a  kinetic (nonequilibrium) 
steady-state solution that corresponds to epitaxial growth through step flow. 
The step stiffness, $\tbe(\theta)$, is determined via perturbations of the kinetic steady state 
for small edge P\'eclet number, $P$, which is the ratio of the deposition to the diffusive flux
along a step edge. In particular, $\tbe$ is found to 
satisfy $\tbe =O(\theta^{-1})$ for 
$O(P^{1/3}) <\theta \ll 1$, which is in agreement with independent, equilibrium-based calculations.
\looseness=-1
\end{abstract}



\pagestyle{myheadings}
\thispagestyle{plain}
\markboth{D. MARGETIS AND R. E. CAFLISCH}{ANISOTROPIC STIFFNESS FROM EPITAXIAL KINETIC MODEL}

\section{Introduction}
\label{sec:intro}

The design and fabrication of novel small devices require the synergy of experiment, mathematical modeling and 
numerical simulation. 
In epitaxial growth, crystal surface features such as thin films, which are building blocks of solid-state devices,
are grown on a substrate by material deposition from above. Despite continued  
progress, the modeling and simulation of epitaxial phenomena remains challenging because it involves
reconciling a wide range of length and time scales.

An elementary process on solid surfaces is the hopping of atoms in the presence of
line defects (``steps'') of atomic height~\cite{evansetal06,jeongwilliams99,williamsbartelt96}: atoms hop on terraces, and attach to and detach from 
step edges (or island boundaries).
Burton, Cabrera and Frank (BCF)~\cite{bcf51} first described each step edge as a boundary moving by mass conservation
of point defects (``adatoms'') which diffuse on terraces. In the BCF theory, the step motion occurs near
thermodynamic equilibrium. Subsequent theories have accounted for far-from-equilibrium processes; for a review 
see section~\ref{sec:model}. 

The macroscale behavior of crystal surfaces is described by use of effective material parameters such as
the step stiffness, $\tbe$~\cite{margetiskohn06}. In principle, $\tbe$
depends on the step edge orientation angle, $\theta$, and is viewed as a quantitative measure of step edge 
fluctuations~\cite{akutsu86,stasevichetal04}. 
Generally, effective step parameters such as $\tbe$ originate from atomistic processes to which inputs are hopping rates for atoms;
in practice, however, the parameters are often provided by phenomenology.
For example, the dependence of $\tbe$ on $\theta$ is usually speculated by invoking the underlying crystal 
symmetry~\cite{banchetal05,hausservoigt05,siegeletal04,spencer04}.\looseness=-1 

In this article we analyze a kinetic model for out-of-equilibrium processes~\cite{caflischetal99,caflischli03}
in order to: (i) derive a Gibbs-Thomson (GT) type formula, which 
relates the adatom flux normal to a curved step edge and the step edge 
curvature~\cite{israelikandel99,jeongwilliams99}; and (ii) determine
the step stiffness $\tbe$, which enters the GT relation, as a function
of $\theta$. For this purpose, we apply perturbations of the kinetic (nonequilibrium) steady state
of the model for small P\'eclet number $P$, which is the ratio of the material deposition 
flux to the diffusive flux along a step edge, i.e. 
\begin{equation}
P = (2 a^3 \bar f) /D_E~, \label{eq:Pdef}
\end{equation}
in which $a$ is an atomic length, $\bar f$ is a characteristic size for the flux $f$ normal 
to the boundary from each side, and $D_E$ is the coefficient for diffusion along the boundary. A factor of $2$ is included in (\ref{eq:Pdef}) since the flux is two-sided and the total flux is of size $2\bar f$.
 For sufficiently small $\theta$ and $P$, we find that the
stiffness has a behavior similar to that predicted by equilibrium-based 
calculations~\cite{stasevich06}.\looseness=-1

For the boundary of a two-dimensional material region, a definition
of $\tbe$ can arise from linear kinetics. In the 
setting of atom attachment-detachment at an edge, 
this theory states that the material flux, $f$, normal to the curved boundary is linear in the difference of the material 
density, $\rho$, at the boundary from a reference or ``equilibrium'' density, $\rho_0$. 
The GT formula connects $\rho_0$ to the boundary curvature, $\kappa$.
For unit layer thickness and negligible step interactions~\cite{jeongwilliams99}, the normal flux reads\looseness=-1
\begin{equation}
f =   \Da (\rho-\rho_0),\label{eq:f-noperm}
\end{equation}
where $\Da$ is the diffusion coefficient for attachment and detachment, and $\rho_0$ is defined by
\begin{equation} 
\rho_0=\rho_*\,e^{\frac{\tbe\,\kappa}{k_BT}}\sim \rho_*\biggl(1+\frac{\tbe}{k_BT}\,
\kappa\biggr),\qquad |\tbe \kappa|\ll k_BT~.
\label{eq:GT}
\end{equation}
The last equation is referred to as the GT 
formula, in accord with standard thermodynamics~\cite{biskupetal04,gibbs28,krishnamacharietal96,landau,rowlinsonwidom82}. 
In~(\ref{eq:GT}), $\rho_*$ is the equilibrium density near a straight step edge 
and $k_B T$ is Boltzmann's energy ($T$ is temperature);
the condition $|\tbe\kappa|\ll k_B T$ is satisfied in most experimental 
situations~\cite{tersoffetal97}. Equation~(\ref{eq:f-noperm}) does not account for
step permeability, by which terrace adatoms hop directly
to adjacent terraces~\cite{ozdemir90,tanakaetal97}. This process is discussed in section~\ref{sec:model}.\looseness=-1

For systems that are nearly in equilibrium, the exponent in~(\ref{eq:GT}) is derived by a thermodynamic 
driving force starting 
from the step line tension $\beta$, the free energy per unit length of the boundary~\cite{gurtin93}. 
The step stiffness $\tbe$ is related to $\beta$ by~\cite{akutsu86,fisher84,fisher82}\looseness=-1
\begin{equation}
\tbe = \beta + \beta_{\theta \theta}\qquad (\beta_\theta:=\partial_\theta\beta)~.
\label{eq:tension-stiff}
\end{equation}
Evidently, the knowledge of $\tbe$ alone does not yield $\beta$ uniquely: by~(\ref{eq:tension-stiff}),
\begin{equation}
\beta(\theta)=C_1\,\cos\theta+C_2\,\sin\theta+\int_0^{\theta}d\vartheta\ \tbe(\vartheta)\,\sin(\theta-\vartheta)
\label{eq:beta-tdbeta}
\end{equation}where $C_1$ and $C_2$ are in principle arbitrary constants.

The parameters $\beta$ and $\tbe$ are 
important in the modeling and numerical simulation of epitaxial phenomena. 
In thermodynamic equilibrium, the angular dependence of the step line tension, $\beta(\theta)$, 
determines the equilibrium (two-dimensional) shape of step edges or islands, e.g. the macroscopic flat parts 
(``facets'') of the step are found by minimizing the step line energy through 
the Wulff construction~\cite{herring51,pengetal99,rottmanwortis84,szalmaetal05,taylor74,wulff1901}.
Near thermodynamic equilibrium,
the step stiffness, $\tbe(\theta)$, controls the temporal decay of fluctuations from 
equilibrium~\cite{akutsu86,jeongwilliams99}. The significance of $\tbe$ was pointed out
by de Gennes in the context of polymer physics almost forty years ago~\cite{deGennes68,einstein03}:
the energy of a polymer (or step edge) can be described by a kinetic energy term proportional
to $\tbe\cdot (dx/dy)^2$, i.e., the stiffness times a ``velocity" squared where $x$ and $y$
are suitable space coordinates and $y$ loosely corresponds to ``time." Starting with a two-dimensional 
Ising model, Stasevich et al.~\cite{stasevich_thesis,stasevich06,stasevichetal04,stasevichetal05}
carried out a direct derivation of $\beta(\theta)$ and $\tbe(\theta)$ from an  
equilibrium perspective based on atomistic key energies. For most systems, however, 
there has been no standard theoretical method for 
determining $\beta(\theta)$ and $\tbe(\theta)$.\looseness=-1

More generally, energetic principles such as a thermodynamic driving force are powerful as a means of describing 
the macroscopic effect of atomistic kinetics. The range of validity of energetic principles is not fully known and is an important unresolved issue.
We believe that energetic arguments should be valid for systems that are nearly in local equilibrium, where  
the relevant processes approximately satisfy detailed balance. For systems that are far 
from equilibrium, however, 
energetic principles may serve as a valuable qualitative guide, even if they are not quantitatively accurate.

The kinetic and atomistic origin of a material parameter that plays the role of the step stiffness are the subject of this article. 
For a step edge or an island boundary on an epitaxial crystal surface, we use the
detailed kinetic model formulated by Caflisch et al.~\cite{caflischetal99,caflischli03} and further developed by Balykov and Voigt \cite{balykovvoigt05,balykovvoigt06} for the dynamics 
of the boundary. The basic ingredients are: (i) diffusion equations for adatom and edge-atom densities
on terraces and along step edges; (ii) a convection equation for the kink density along step edges;
and (iii) constitutive, algebraic laws for adatom fluxes, sources for kinks and
the kink velocity by mean-field theory. 
This model admits a kinetic (nonequilibrium) steady state that allows for epitaxial growth via step flow.
The model has been partly validated by kinetic Monte Carlo simulations~\cite{caflischetal99}.

The detailed step model described in~\cite{caflischetal99,caflischli03} and 
section~\ref{subsec:noneq-kin} focuses on the
kinetics of adatoms, edge-atoms and kinks at a step edge. As discussed by
Kallunki and Krug~\cite{kallunkikrug03}, an edge-atom is {\it energetically} equivalent to two
kinks. For example, the equilibrium density of kinks is proportional to 
$\exp[-\varepsilon/(k_B T)]$ while the equilibrium density of edge-atoms is proportional to
$\exp[-2\varepsilon/(k_B T)]$, in which $\varepsilon$ is defined as the kink
energy in~\cite{kallunkikrug03}, or identified with $-(k_B T/2)\log({D_K/D_E})$ 
in~\cite{caflischetal99}; $D_K$ and $D_E$ are diffusion coefficients
for kinks and edge-atoms. On the
other hand, the {\it kinetics} in~\cite{caflischetal99,caflischli03}
are different for edge-atoms and kinks, since edge-atoms can hop 
at rate $D_E$, while kinks move through detachment of atoms at
rate $D_K$. This situation is consistent with the kinetics described in~\cite{kallunkikrug03}, in which
the $D_E$ and $D_K$ are proportional to $\exp[-E_{st}/(k_B T)]$ and
$\exp[-E_{det}/(k_B T)]$, respectively.

We are aware that the mean-field laws applied here, although plausible and analytically tractable,
pose a limitation: actual systems are characterized by atomic correlations, which can cause
deviations from this mean-field approximation. In particular, the validity of the mean-field assumption 
may be limited to orientation angles $\theta$ in some neighborhood of $\theta=0$. 
Note also that the most interesting results of this analysis are for $\theta$ near zero.
Determination of the range of validity for this model is an important endeavor but beyond the scope of this paper. An extension of this model, 
which could improve its range of validity, would be to explicitly track the kinks in a step edge. This additional discreteness in the model 
would make the analysis of step stiffness more difficult.  Our analysis is a systematic study 
of predictions from the mean-field approach only, and the conclusions presented here are all derived within the context of this approach. 
On the other hand, our analysis is more detailed than previous treatments of step stiffness, since it is based on kinetics rather than 
a thermodynamic driving force. Moreover, the model includes atomistic information, through a
density of adatoms, edge-atoms and kinks.  

For evolution near the kinetic steady state, we derive for the mass flux, $f$, a term analogous to the Gibbs-Thomson
formula~(\ref{eq:GT}), and subsequently find the corresponding angular dependence of the step stiffness, $\tbe(\theta)$. 
Our main assumptions are: 
(i) the motion of step edges or island boundaries is slower than the diffusion of adatoms and edge-atoms
and the convection of kinks, which amounts
to the ``quasi-steady approximation'';
(ii) the mean step edge radius of curvature, $\kappa^{-1}$, is large compared
to other length scales including  the step height, $a$;
and (iii) the edge P\'eclet number, $P$, given by~(\ref{eq:Pdef}) is sufficiently small, which signifies the usual regime for
molecular beam epitaxy (MBE).
To the best of our knowledge, the analysis in this paper offers the first kinetic derivation of a
Gibbs-Thomson type relation and the step stiffness for all admissible values of the step edge orientation angle, $\theta$.
(This approach is distinctly different from the one in e.g.~\cite{shenoy04}
where classical elasticity is invoked.) Our results for the stiffness are summarized in section~\ref{sec:res}; 
see~(\ref{eq:stiff-smalltheta})--(\ref{eq:k0-asymp1}). 

A principal result of our analysis is that $\tbe =O(\theta^{-1})$ for 
$O(P^{1/3}) < \theta \ll 1$, which by~(\ref{eq:beta-tdbeta}) yields $\beta=O(\theta\ln\theta)$ for the step line tension.
This result is in agreement with the independent analysis 
in~\cite{stasevich_thesis,stasevich06,stasevichetal04,stasevichetal05}, which makes use
of equilibrium concepts. A detailed comparison of the two approaches is not addressed in our analysis. 
Our findings are expected to 
have significance for epitaxial islands, for example in predicting their facets, their roughness (e.g., fractal or 
smooth island boundaries) and their stability, as well as
for the numerical simulation of epitaxial growth. More generally, our analysis can serve as a guide 
for kinetic derivations of the GT relation in other material systems. 
For example, it should be possible to derive the step stiffness for a step in local thermodynamic equilibrium
within the context of the same model. This topic is discussed briefly in 
section~\ref{sec:conclusion}. \looseness=-1
 
The present work extends an earlier analysis by Caflisch and Li~\cite{caflischli03},
which addressed the stability of step edge models and the derivation of the GT relation. 
The analysis in~\cite{caflischli03}, however, 
only determined the value of $\tbe$ along the high-symmetry orientation, $\theta=0$. This restriction was due to a  
scaling regime used in~\cite{caflischli03} on the basis of mathematical rather than physical principles.
In the present article we transcend the analytical limitations of~\cite{caflischli03} by applying
perturbation theory guided by the physics of the step-edge evolution near the kinetic steady state.
\looseness=-1

Our analysis also leads to formulas for {\it kinetic rates} in boundary conditions involving adatom fluxes.
In particular, the attachment-detachment rates are derived as functions
of the step edge orientation, and are shown to be different for up- and down-step edges.
This asymmetry amounts to an Ehrlich-Schwoebel (ES) effect~\cite{ehrlichhudda66,schwoebelshipsey66}, 
due to geometric effects rather than a difference in energy barriers.
In addition, if the terrace adatom densities are treated as input parameters,
the adatom fluxes involve effective {\it permeability} rates,
by which a fraction of adatoms directly hop to adjacent terraces (without
attaching to or detaching from step edges)~\cite{filimonovhervieu04,ozdemir90,tanakaetal97}. 
Our main results for the kinetic
rates are described by~(\ref{eq:ES-Dpm})--(\ref{eq:Apm-def}).

In this article we do not address the effects of elasticity, which are due for instance to bulk stress.
One reason is that elasticity requires a non-trivial modification of the kinetic
model that we use here. This task lies beyond our present scope. 
Another reason is that, in many physically interesting situations, the influence of elasticity may be described 
well via
long-range step-step interactions that do not affect the step stiffness. The study of elastic effects
is the subject of work in progress.\looseness=-1 

The remainder of this article is organized as follows. 
In section~\ref{sec:model} we review the relevant island 
dynamics model and the concept of step stiffness: In section~\ref{subsec:geom} we 
introduce the step geometry; in section~\ref{subsec:bcf} we outline elements of the BCF model, 
which highlight the GT formula; in section~\ref{subsec:noneq-kin} we describe the previous kinetic, 
nonequilibrium step-edge
model~\cite{caflischetal99,caflischli03}, which is slightly revised here; and
in section~\ref{subsec:stiff-calc} we outline our program for the stiffness, based on
the perturbed kinetic steady state for small step edge curvature, $\kappa$. In 
section~\ref{sec:res} we provide a summary of our main results.
In section~\ref{sec:steady} we derive analytic formulas pertaining to the kinetic steady state:
In section~\ref{subsec:flux} we use the mass fluxes as inputs and derive the ES
effect~\cite{ehrlichhudda66,schwoebelshipsey66}; and in section~\ref{subsec:density} we use the
mass densities as inputs to derive asymmetric, $\theta$-dependent step-edge permeability rates.
In section~\ref{sec:GT-stiff} we apply perturbation theory 
to find $\tbe(\theta)$ by using primarily the mass fluxes as inputs: 
In section~\ref{subsec:pertb} we carry out the perturbation analysis to first order for 
the edge-atom and kink densities as $\kappa\to 0$; in section~\ref{subsec:stiff} we derive the
step stiffness as a function of $\theta$; and in section~\ref{subsec:extension}
we discuss an alternative viewpoint on the stiffness.
In section~\ref{sec:conclusion} we discuss our results, and outline possible limitations.
The appendices provide derivations and proofs needed in the main text.

\section{Background}
\label{sec:model}

In this section we provide the necessary background for the derivation of the step stiffness.
First, we describe the step configuration.
Second, we revisit briefly the constituents of the BCF theory with focus on 
the GT formula and the step stiffness, $\tbe$. Our review provides the introduction of $\tbe$ from a kinetic rather
than a thermodynamic perspective. Third, we describe in detail the nonequilibrium kinetic 
model~\cite{caflischetal99,caflischli03} with emphasis on the mean-field constitutive laws 
for edge-atom and kink densities.
Fourth, we set a perturbation framework for the derivation of $\tbe(\theta)$. 

\subsection{Step geometry and conventions}
\label{subsec:geom}

Following~\cite{caflischetal99,caflischli03} we consider a simple cubic crystal (solid-on-solid model)
with lattice spacing $a$ and
crystallographic directions identified with the $x$, $y$ and $z$ axes of the Cartesian system. 
The analysis of this paper is for a step edge or island boundary to which there is flux $f$ of atoms from the adjoining terraces. 
The flux $f$ may vary along the edge, as well as in time, and it comes from both sides of the edge, but it is characterized by 
a typical size $\bar f$ which has units of $(length \cdot time)^{-1}$.  In~\cite{caflischetal99,caflischli03} the geometry was 
specialized to a  step train with interstep distance $2L$ and deposition flux $F$, so that in steady state the  flux to the step is $f=LF$. 
This global scenario is not necessary, however, since the analysis here is local and only  requires a nonzero quasi-steady flux $f$. This could
occur even with no deposition flux $F=0$; for example, in annealing. 

\begin{figure}
\begin{tabular}{cc}
\includegraphics[width=2.3in,height=3.4in]{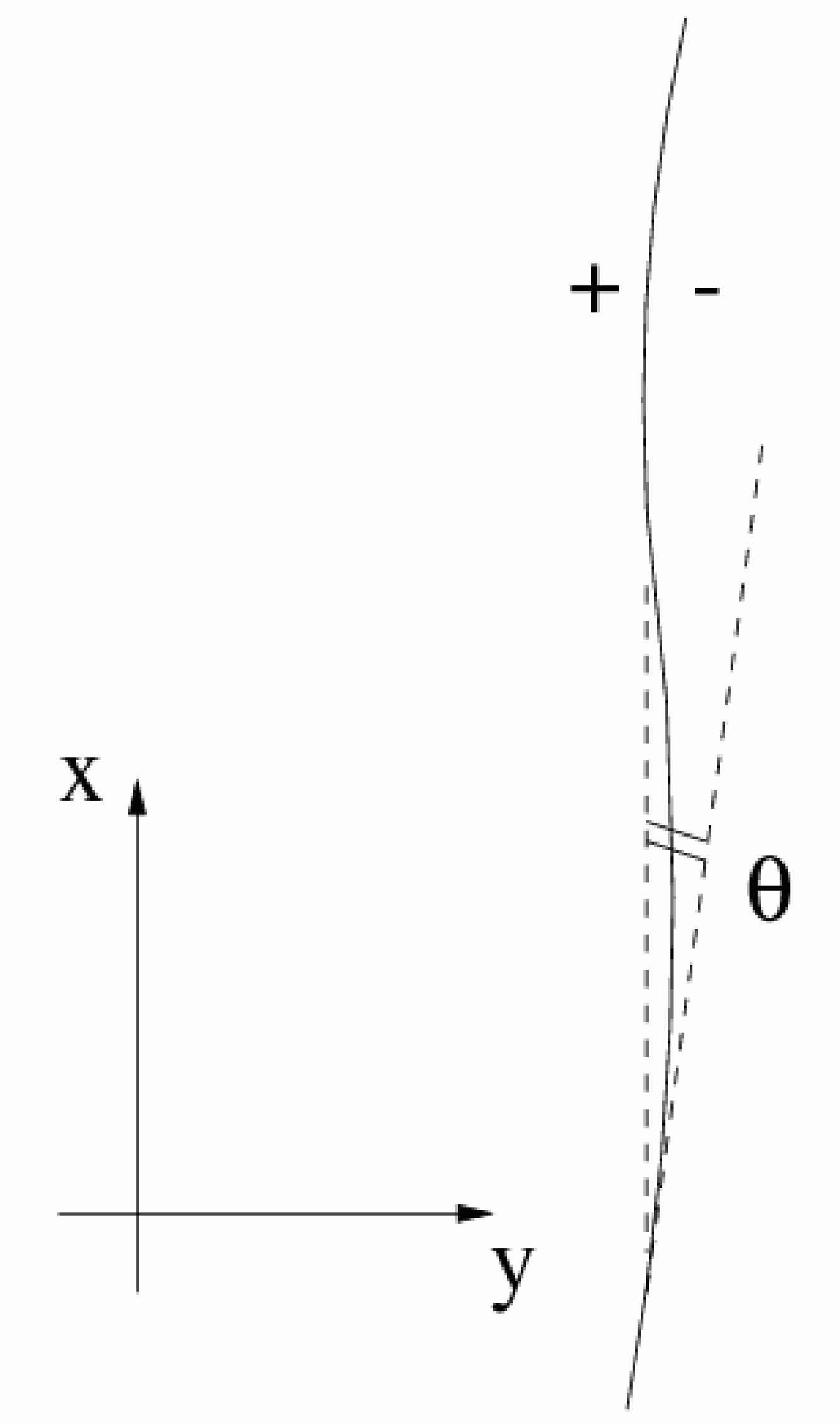} &
\includegraphics[width=4.0in,height=4.5in,angle=0]{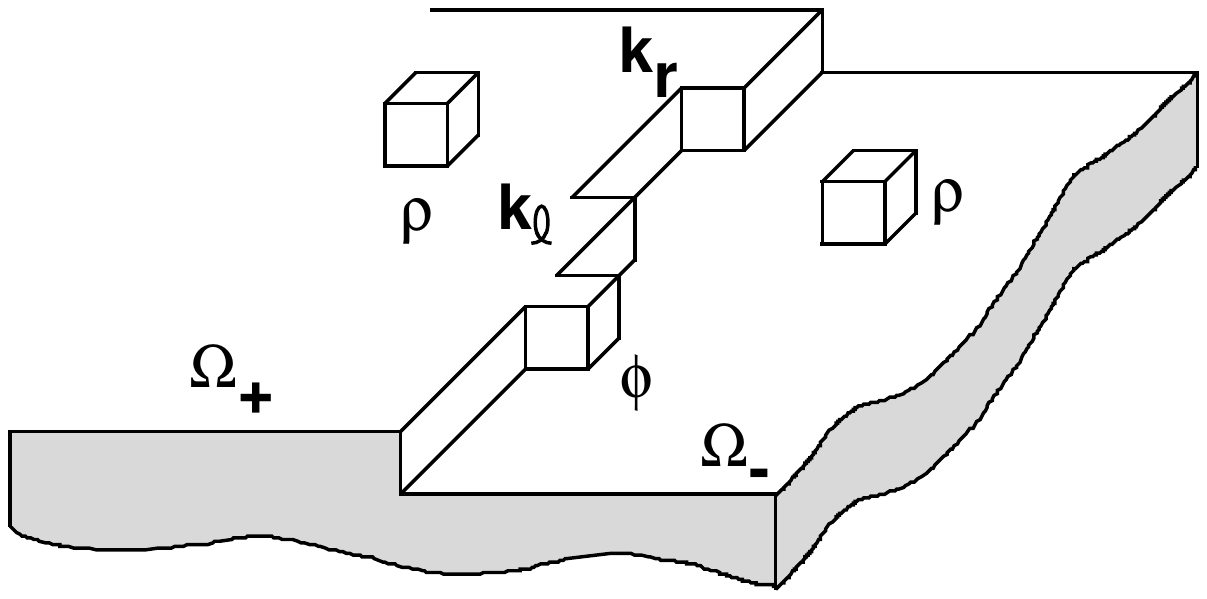}
\end{tabular}
\caption{The macroscopic (left) and microscopic (right) views of a step edge 
in the (high-symmetry) $xy$-plane of a crystal. 
In the macroscopic view, the step edge orientation relative to the $x$ axis is indicated by the angle $\theta$.
The $+$ ($-$) sign indicates an upper (lower) terrace. The surface height decreases to the right. 
The microscopic view shows adatoms ($\rho$), edge-atoms ($\phi$), left-facing kinks 
($k_\ell$) and right-facing kinks ($k_r$); $\Omega_{+}$ ($\Omega_{-}$) is the
region of the upper (lower) terrace.\looseness=-1}\label{fig:Fig1}
\end{figure}

For algebraic convenience we adopt and extend the notation conventions of~\cite{caflischli03}. 
Specifically, we use the following symbols: ($x, y, z$) for dimensional spatial 
coordinates, $t$ for time, $D$ for any diffusion coefficient, $\rho$ for number density
per area , and $\xi$ for number density per length; and define the corresponding nondimensional quantities 
$\tilde x, \tilde y, \tilde z, \tilde t, \tilde D, \tilde \rho, \tilde \xi$ by
\begin{eqnarray}
(\tilde x, \tilde y, \tilde z) &:=& (x/a, y/a, z/a)~,
\label{eq:nond-space}\\
\tilde t &:=& (a\bar f)\, t~,
\label{eq:nond-time}\\
\tilde D &:=& D/(a^3\bar f)~,
\label{eq:nond-diff}\\
\tilde \rho &:=& a^2\rho~,
\label{eq:nond-rho}\\
\tilde \xi &:=& a\xi~.
\label{eq:nond-phi}
\end{eqnarray}
Now drop the tildes, so that $x, y, z, t, D, \rho, \xi$ are dimensionless.
This choice amounts to measuring all distances in units of $a$ and all times in units
of $(a\bar f)^{-1}$. Equivalently, (\ref{eq:nond-space})--(\ref{eq:nond-phi}) correspond to setting 
$a=1$ and $\bar f=1$. For our analysis, the single most important dimensionless parameter is
the Peclet number $P$ from (\ref{eq:Pdef}), which is equal to $2 D_E^{-1}$ after nondimensionalization; i.e.
\begin{equation}
D_E=2 P^{-1}.
\end{equation}

Next, we describe the coordinates of the step geometry in more detail. 
We consider step boundaries that stem
from perturbing a straight step edge coinciding with a fixed axis (e.g., the $x$-axis). 
All steps are parallel to the high-symmetry (``basal''), 
$xy$-plane of the crystal. The projection of each edge on the basal plane is represented macroscopically by a smooth curve 
with a local tangent that forms the (signed) angle $\theta$ with the $x$-axis, 
where $-\theta_0< \theta <\theta_0$~\footnote{The definition of $\theta$
here is the same as that in~\cite{caflischetal99}, but different from the one in~\cite{caflischli03} 
where $\theta$ is the angle formed by the local tangent and the $y$ axis.}.
Without loss of generality we take $0\le\theta< \theta_0$ and assume that $\theta_0 < \pi/4$ in our analysis. 
We take the upper terrace to be to the left of an edge so that
all steps move to the right during the growth process. 
So, the projection of each step edge is represented by\looseness=-1
\begin{equation}
y=Y(x,t)~,
\label{eq:edge-repr}
\end{equation}
where $Y(x,t)$ is a sufficiently differentiable function of ($x, t$). 

It follows that the unit normal and tangential vectors to the step boundary are~\cite{caflischli03}
\begin{equation}
\norm = (\sin\theta,-\cos\theta)=(y_s,-x_s)~,\qquad \tang=(\cos\theta,\sin\theta)=(x_s,y_s)~,
\label{eq:unit-def}
\end{equation}
where $s$ is the arc length and lowercase subscripts denote 
partial differentiation (e.g., $x_s:=\partial_s x$) unless it is noted or implied otherwise. 
The step edge curvature is 
\begin{equation}
\kappa=-\theta_s~.
\label{eq:kappa}
\end{equation} 

There is one more geometric relation that deserves attention.
By denoting the densities of
left- and right-facing kinks $k_{l}$ and $k_{r}$, respectively, we have~\cite{caflischetal99} 
\begin{equation}
k_r-k_l=-\tan\theta~;
\label{eq:krl-theta}
\end{equation}
see section~\ref{subsec:noneq-kin} for further discussion. 
This geometric relation poses a constraint on the total kink density, $k$ ($k\ge 0$). By 
\begin{equation}
k:=k_{r}+k_{l}\label{eq:k-def}
\end{equation}and~(\ref{eq:krl-theta}), $k$ must satisfy
\begin{equation}
k\ge |\tan\theta|~.
\label{eq:k-ctr}
\end{equation}

The formulation of a nonequilibrium kinetic step edge model (section~\ref{subsec:noneq-kin}) 
requires the use of several coordinate systems for an island boundary; these are described in appendix~\ref{app:coods}.
In the following analysis it becomes advantageous to use $\theta$ as the main local coordinate. 
Its importance as a dynamic variable along a step edge is implied by the steady-state limit
$k\to |\tan\theta|$ as $\kappa\to 0$ and $P\to 0$; see~(\ref{eq:k0-asymp1-corr}).
Some useful identities that enable transformations to the $(\theta, t)$ variables are 
provided in appendix~\ref{app:coods}.\looseness=-1

\subsection{BCF model}
\label{subsec:bcf}

In the standard BCF theory~\cite{bcf51} the projection of step edges on the basal plane are smooth
curves that move by the attachment and detachment of atoms due to mass conservation.
The BCF model comprises the following near-equilibrium evolution laws.
(i) The adatom density solves the diffusion equation on terraces.
(ii) The adatom flux and density satisfy (kinetic) boundary conditions 
for atom attachment-detachment at step edges. 
(iii) The step velocity equals the sum of the adatom fluxes normal to the edge. 
In this setting, the GT formula links
the normal mass flux to the step edge curvature.\looseness=-1 

We next describe the equations of motion in the BCF model for 
comparisons with the kinetic model of section~\ref{subsec:noneq-kin}.
The density, $\rho$, of adatoms on each terrace solves
\begin{equation}
\partial_t \rho-\Dt\,\Delta\rho=F~,
\label{eq:rho-pde}
\end{equation}where $\Dt$ is the terrace diffusion coefficient and $\Delta$ denotes the Laplacian in $(x,y)$.
\looseness=-1 

As an extension of the BCF model, the boundary conditions for~(\ref{eq:rho-pde}) are now formulated by linear kinetics with
inclusion of both atom attachment-detachment {\it and} step 
permeability~\cite{jeongwilliams99,ozdemir90,tanakaetal97}:
\begin{equation}
f_{\pm}=\Da^{\pm}\, (\rho_{\pm}-\rho_0^\pm)\pm D_p^\pm\,(\rho_+-\rho_-)~;
\label{eq:f-bcf}
\end{equation}
cf.~(\ref{eq:f-noperm}). Here, $f_\pm$ is the adatom flux normal to an edge from the 
upper ($+$) or lower ($-$) terrace, i.e.,
\begin{equation}
\mp f_\pm:=v\rho_\pm+\Dt\norm\cdot(\nabla\rho)_\pm~,\label{eq:fpm-def}
\end{equation}
$\rho_{\pm}$ is the terrace adatom density restricted to the step edge,
$\Da^{\pm}$ is the attachment-detachment rate coefficient and $D_p^{\pm}$
is the permeability rate coefficient. These rates can account for different up- and down-step 
energy barriers, e.g. the ES effect in the case of $\Da^\pm$~\cite{ehrlichhudda66,schwoebelshipsey66}. 
The reference density $\rho_0^\pm$ is 
given by~(\ref{eq:GT}) where $\rho_*$
is replaced by $\rho_*^\pm$ for up- and down-step edge asymmetry. 
Evidently,~(\ref{eq:f-bcf}) forms an extension of formula~(\ref{eq:f-noperm}) but still
corresponds to near-equilibrium kinetics; it will be modified in section~\ref{subsec:noneq-kin}. 

Equations~(\ref{eq:rho-pde}) and~(\ref{eq:f-bcf}) provide the fluxes
$f_\pm$ as functions of the step edge position and curvature. The step velocity, $v$, is then determined by
mass conservation,
\begin{equation}
v=f_+ + f_-~.\label{eq:st-vel}
\end{equation}In this formulation, step-edge diffusion and kink motion are neglected. 
In the next section, the BCF model is enriched with kinetic boundary conditions that account for the
motion of edge-atoms and kinks.

\subsection{Atomistic, nonequilibrium kinetic model}
\label{subsec:noneq-kin}

In this section we revisit the kinetic model by Caflisch et al.~\cite{caflischetal99,caflischli03}, which 
is an extension of the BCF model (section~\ref{subsec:bcf}) to nonequilibrium processes. We
apply this kinetic model~\cite{caflischetal99} to step edges of arbitrary orientation;
and further revise it to account for a step edge diffusion coefficient defined 
along the (fixed) crystallographic $x$-axis.
This last feature, although not important for our present purpose of calculating the step stiffness, renders
the model consistent with recent studies of the edge-atom migration along a step edge~\cite{kallunkikrug03}.
The following processes are included. (i) Adatom diffusion on terraces, which is described
by~(\ref{eq:rho-pde}) of the BCF theory, and edge-atom diffusion along step edges. (ii) Convection of kinks on step edges with sinks and sources
to account for conversion of terrace adatoms and edge-atoms to kinks. 
(iii) Constitutive laws that relate mass fluxes, sources for kinks and the step velocity 
with densities via a mean-field theory, and modify the BCF laws~(\ref{eq:f-bcf}) and~(\ref{eq:st-vel}). 
In this model, kink densities are assumed sufficiently
small, enabling the neglect of higher-order terms within the mean-field approach. Recently,
extensions of this theory were developed~\cite{balykovvoigt05,balykovvoigt06,filimonovhervieu04},
including higher kink densities by Balykov and Voigt~\cite{balykovvoigt05,balykovvoigt06}.
Next, we state the requisite equations of motion in addition to~(\ref{eq:rho-pde}) for adatom terrace diffusion.\looseness=-1

\subsubsection{Equations of motion along step edges}
\label{sssec:eq-mot}

An assumption inherent to the present model is the different kinetics 
of kinks and edge-atoms. Each of these species is of course not
conserved separately, since edge-atoms can generate kinks, but 
can be described by a distinct density:
$\phi(x,t)$ for edge-atoms and $k(x,t)$ for kinks. In addition, 
their motion is different: the edge-atom flux follows from
gradients of the density $\phi$; while the kink flux stems from a velocity field, $w$. 

We proceed to describe the equations of motion.
The edge-atom number density, $\phi(x,t)$, solves
\begin{equation}
\partial_t\phi -\De\,\partial_x^2\phi=\frac{\fp+\fm}{\cos\theta}-f_0~,
\label{eq:phi-pde}
\end{equation}
where $\De$ is the step edge diffusivity defined along the high-symmetry ($x$-) axis
and $f_0$ represents the loss of edge-atoms to kinks; see~(\ref{eq:fpm}) and~(\ref{eq:f0}) below. 
For later algebraic convenience, it
is advantageous to transform~(\ref{eq:phi-pde}) to $(\theta, t)$ variables. By the formulas
~(\ref{eq:partx})  and~(\ref{eq:part-tx}) 
of appendix~\ref{app:coods}, (\ref{eq:phi-pde}) is thus recast to 
\begin{equation}
\partial_t|_\theta\phi+\kappa(v_\theta+v\tan\theta)\partial_\theta\phi-\De\frac{\kappa}{\cos\theta}
\partial_\theta\frac{\kappa}{\cos\theta}\partial_\theta\phi=\frac{\fp+\fm}{\cos\theta}-f_0~.
\label{eq:phi-pde-th}
\end{equation}

We turn our attention to kinks. The total kink density, $k(x,t)$, of~(\ref{eq:k-def}) solves
\begin{equation}
\partial_t k+\partial_x[w(k_r-k_l)]=2(g-h)~,
\label{eq:k-pde}
\end{equation}
where $w(k_r-k_l)=-w\tan\theta$ is the 
flux of kinks with respect to the $x$-axis, $g$ is the net gain in kink pairs
due to nucleation and breakup, and $h$ is the net loss in kink
pairs due to creation and annihilation~\cite{caflischetal99}. 
The terms $w$, $g$ and $h$ are described as functions of densities 
in~(\ref{eq:w-def})--(\ref{eq:h-def}) below.
In the $(\theta,t)$ coordinates, (\ref{eq:k-pde}) reads
\begin{equation}
\partial_t|_\theta k+\kappa(v_\theta+v\tan\theta)\partial_\theta k+\frac{\kappa}{\cos\theta}\partial_\theta(w\tan\theta)=2(g-h)~.
\label{eq:k-pde-th}
\end{equation}

Equations~(\ref{eq:phi-pde}) and~(\ref{eq:k-pde}) can be transformed to other coordinates,
including the $(s, t)$ variables where $s$ is the arc length. 
For completeness, in appendix~\ref{app:id-edge}
we provide relations that are needed in such transformations; and in appendix~\ref{app:eq-mot}
we describe the ensuing equations of motion in the $(s,t)$ coordinates.

Partial differential equations~(\ref{eq:phi-pde}) and (\ref{eq:k-pde}) are
coupled with the motion of step edges. In the following analysis, we apply the quasi-steady approximation,
neglecting the time derivative in~(\ref{eq:phi-pde-th}) and~(\ref{eq:k-pde-th}).
For definiteness, the boundary conditions in $x$ 
can be taken to be periodic. It remains to prescribe boundary conditions for atom attachment-detachment,
i.e., specify $f_\pm$ in~(\ref{eq:fpm-def}).
 In the present nonequilibrium context, $f_\pm$ are no longer given by~(\ref{eq:f-bcf}) of the BCF model, 
as discussed next.

\subsubsection{Constitutive laws}
\label{sssec:mf}

Following~\cite{caflischetal99,caflischli03} we describe mean-field constitutive laws for fluxes related to a tilted step 
edge (at $\theta\neq 0$). We also provide a geometric relation for the step edge velocity, $v$, which in a certain sense
replaces the BCF law~(\ref{eq:st-vel}). Because the explanations are given 
elsewhere~\cite{balykovvoigt05,caflischetal99},
we state the mean-field laws without a detailed discussion of their origin. 

By mean-field theory, the terrace adatom flux normal to the step edge is~\cite{caflischetal99}
\begin{eqnarray}
f_\pm&=&[\Dt\rho_\pm-\De\phi+l_{j_\pm}(\Dt\rho_\pm-D_K)k+m_{j_\pm}(\Dt\rho_\pm\phi-D_Kk_rk_l)\nonumber\\
\mbox{}&&+n_{j_\pm}(\Dt\rho_\pm k_rk_l-D_B)]\cos\theta,\quad j_+=2,\ j_-=3~,
\label{eq:fpm}
\end{eqnarray}where $l_j$, $m_j$ and $n_j$ are (effective) coordination numbers (positive integers) that count
the number of possible paths in the kinetic processes, weighted by the relative probability of a particle to be at the corresponding position. Also, $D_K$ is the diffusion
coefficient for an atom from a kink, and $D_B$ is the diffusion coefficient for an atom from
a straight edge. 
By neglect of $D_K$ and $D_B$, (\ref{eq:fpm}) readily becomes
\begin{equation}
f_\pm=(1+l_{j_\pm} k+m_{j_\pm}\phi+n_{j_\pm} k_rk_l)\Dt\rho_\pm\cos\theta-D_E\phi\cos\theta~.
\label{eq:fpm-simp}
\end{equation}
Omitting $D_K$ and $D_B$ is inconsistent with detailed balance, but has little effect on
the kinetic solutions described below.

Similarly, the mean-field kink velocity reads~\cite{caflischetal99}
\begin{equation}
w= l_1\De\phi+\Dt(l_2\rho_++l_3\rho_-)-\ls D_K\sim l_1\De\phi+\Dt(l_2\rho_++l_3\rho_-)~.
\label{eq:w-def}
\end{equation}
The gain in kink pairs from nucleation and breakup involving an edge-atom is~\cite{caflischetal99}
\begin{eqnarray}
g&=&\phi(m_1\De\phi+m_2\Dt\rho_++m_3\Dt\rho_-)-\ms\,D_K k_rk_l\nonumber\\
 &\sim& \phi(m_1\De\phi+m_2\Dt\rho_++m_3\Dt\rho_-)~.
\label{eq:g-def}
\end{eqnarray}
The respective loss of kink pairs by atom attachment-detachment is~\cite{caflischetal99} 
\begin{eqnarray}
h&=&(n_1\De\phi+n_2\Dt\rho_++n_3\Dt\rho_-)k_rk_l-\ns\,D_B\nonumber\\
&\sim& (n_1\De\phi+n_2\Dt\rho_++n_3\Dt\rho_-)k_rk_l~.
\label{eq:h-def}
\end{eqnarray}In the above,
\begin{equation}
p_{ij}:=p_i+p_j,\quad p_{ijk}:=p_i+p_j+p_k;\qquad p=m,\,n,\,l~.
\label{eq:sums-def}
\end{equation}

The constitutive laws are complemented by 
\begin{equation}
f_0=wk+2g+h~,
\label{eq:f0}
\end{equation}
which enters~(\ref{eq:phi-pde-th}). The step edge velocity, $v$, stems from a geometric relation;
see appendix~\ref{app:vel} for details. Specifically,
\begin{equation}
v=\frac{f_0}{1+\phi\kappa\cos\theta}\cos\theta=\frac{wk+2g+h}{1+\phi\kappa\cos\theta}\cos\theta~.
\label{eq:v-def}
\end{equation}

\subsection{Program for step stiffness}
\label{subsec:stiff-calc}

In this section we delineate a program for the calculation of the step stiffness from
the model of section~\ref{subsec:noneq-kin}. The key idea is to reduce
the nonequilibrium law~(\ref{eq:fpm}) to the linear kinetic law~(\ref{eq:f-bcf}) 
by treating the normal fluxes, $f_\pm$, 
as external, free to vary, $O(1)$ parameters of the equations of motion
along a step edge. In this context, the diffusion
equation~(\ref{eq:rho-pde}) is not invoked. Our method relies on
the perturbation of a solution for the densities $\phi$ and $k$. The solution
studied here is that of the kinetic steady state, under the assumption that it can be reached. 
Accordingly, we neglect the time derivative in the zeroth-order equations of motion; furthermore,
we neglect this derivative to the next higher order by imposing the quasi-steady approximation.
Another case, left for
future work, is that of thermodynamic equilibrium; see section~\ref{sec:conclusion}. In summary, 
we apply the following procedure:\looseness=-1

(i) To extract the kinetic steady state, we set $\partial_t|_\theta \equiv 0$ 
and $\kappa=0$ (i.e., we consider straight edges). This leads to
a system of algebraic equations for $(\phi, k)\equiv (\phi^{(0)}, k^{(0)})$~\footnote{In this context,
the superscript in parentheses denotes the perturbation order in $\kappa$.}. The
coefficients of this system depend on $\theta$ and $f_\pm$. 
In principle, $(\phi^{(0)}, k^{(0)})$ 
cannot be found in simple closed form at this stage.

(ii) We assume that $P\ll 1$, and determine relatively simple
expansions for $(\phi^{(0)}, k^{(0)})$ in powers of $P$ for
$0\le\theta< O(P^{1/3})$ and $O(P^{1/3})<\theta<\pi/4$. 

(iii) We replace $(\phi, k)$ by $(\phi^{(0)}, k^{(0)})$ in the constitutive law~(\ref{eq:fpm})
and compare the result to~(\ref{eq:f-bcf}). Here, our analysis follows up two mathematically
equivalent but physically distinct routes. (a) By taking $f_\pm$
as input parameters, we derive formulas for the adatom reference densities, $\rho_*^{\pm}$,
and attachment-detachment rates, $\Da^\pm$, that depend on $f_\pm$; cf.~(\ref{eq:f-bcf}). Step permeability
is not manifested in this setting ($D_p\equiv 0$). (b) By considering $\rho_\pm$ as inputs, we predict
attachment-detachment rates and non-vanishing step permeability rates.

(iv) We consider perturbations of the kinetic steady state by taking $0<|\kappa|\ll 1$,
i.e. slightly curved step edges. Accordingly, we let
\begin{equation}
\phi\sim \phi^{(0)} +\phi^{(1)}\,\kappa~,\quad k\sim k^{(0)}+k^{(1)}\,\kappa~,
\label{eq:phk-app}
\end{equation}
where $\kappa\phi^{(1)}$ and $\kappa k^{(1)}$ are deviations from the kinetic steady state
and depend on $(\phi^{(0)}, k^{(0)})$. Expansion~(\ref{eq:phk-app})
is imposed on physical rather than mathematical grounds. Indeed, if the mean-field flux~(\ref{eq:fpm})
is expected to reduce to the linear kinetic law~(\ref{eq:f-bcf}),
then $\phi$ must be linear in $\kappa$. The equations of motion along an edge and the constitutive
laws are linearized in $\kappa\phi^{(1)}$ and $\kappa k^{(1)}$.

(v) By treating $f_\pm$ as input external parameters, we replace $\phi$ and $k$ in the right-hand
side of the constitutive
law~(\ref{eq:fpm-simp}) by expansions~(\ref{eq:phk-app}). 
Subsequently, we determine the stiffness $\tbe(\theta;\fp,\fm)$ by comparison 
to~(\ref{eq:f-bcf}) in view of~(\ref{eq:GT}).

The choice of fluxes $f_\pm$ or densities $\rho_\pm$ as input parameters is a physics modeling question. Although the mathematical results are equivalent for the two choice,  the physical
interpretation of these results is different, as stated above.

\section{Main results}
\label{sec:res}

Here, we give the main formulas stemming from our analysis of the
kinetic model described in section~\ref{subsec:noneq-kin}. 
A necessary condition for our perturbation analysis is $0\le \kappa<O(P)\ll 1$,
to be shown via a plausibility argument in section~\ref{subsec:pertb}.
Derivations and other related details are provided in sections~\ref{sec:steady} and~\ref{sec:GT-stiff}.

\subsection{\em ES effect (section~$\ref{subsec:flux}$)} \label{ESeffect}
When the fluxes $f_\pm$ are input parameters, the attachment-detachment of adatoms
from a terrace to an edge is {\it asymmetric}. So, the related diffusion coefficients $\Da^\pm$,
or attachment and detachment kinetic rates, 
which enter~(\ref{eq:f-bcf}), are found to be
different for an upper and lower terrace:
\begin{eqnarray}
\Da^+&=&\Dt\bigl[1+l_2k^{(0)}+m_2\phi^{(0)}+\textstyle{\frac{1}{4}}n_2({k^{(0)}}^2-\tan^2\theta)\bigr]\cos\theta~,\nonumber\\
\Da^-&=&\Dt\bigl[1+l_3k^{(0)}+m_3\phi^{(0)}+\textstyle{\frac{1}{4}}n_3({k^{(0)}}^2-\tan^2\theta)\bigr]\cos\theta~,
\label{eq:ES-Dpm}
\end{eqnarray}where $0\le\theta<\pi/4$ and $(l_2,m_2,n_2)\neq (l_3, m_3, n_3)$.
For $0<P\ll 1$, we show that~(\ref{eq:ES-Dpm}) reduce to
\begin{equation}
\Da^\pm \sim \Dt (1+l_{j_\pm}\tan\theta)\cos\theta~,
\label{eq:Dpm-unif}
\end{equation}where $j_+=2$ and $j_-=3$. In this description, there is no step permeability.
Note that the results presented in this section and their derivations do not depend on the
step edge curvature.

\subsection{\em Step permeability (section~$\ref{subsec:density}$)} \label{stepPermeability}
By using the adatom densities
$\rho_\pm$ as input external parameters, we show that step permeability coexists
with the ES effect; cf.~(\ref{eq:f-bcf}). 
For $O(P^{1/3})<\theta<\pi/4$ the diffusion coefficients for permeability are
\begin{equation}
D_p^\pm = \Dt\,\frac{A_\mp (1+l_{j_\pm}\tan\theta)}{1+(A_+ +A_-)\cos\theta}\cos^2\theta~.
\label{eq:perm-Dpm}
\end{equation}
The accompanying (asymmetric) attachment-detachment diffusion coefficients are
\begin{equation}
\Da^\pm=\Dt \frac{1+l_{j_\pm}\tan\theta\pm  A_\mp (l_2-l_3)\sin\theta}{1+(A_+ +A_-)\cos\theta}\cos\theta~,
\label{eq:Dpm-ES-dens}
\end{equation}
where
\begin{equation}
A_+=\frac{1}{\sin\theta}\,\frac{1+l_3\tan\theta}{Q({\bf l})}~,\qquad 
A_-=\frac{1}{\sin\theta}\,\frac{1+l_2\tan\theta}{Q({\bf l})}~,
\label{eq:Apm-def}
\end{equation}
\begin{equation}
Q({\bf p})=p_1(1+l_2\tan\theta)(1+l_3\tan\theta)+p_2(1+l_3\tan\theta)+p_3(1+l_2\tan\theta)~,\label{eq:Q-def}
\end{equation}
with ${\bf p}:=(p_1,\,p_2,\,p_3)$ and $p=l,\,m,\,n$; in~(\ref{eq:Apm-def}), ${\bf l}=(l_1,\,l_2,\,l_3)$.
Note that $D_p^\pm$ and $\Da^\pm$ here are independent of $f_\pm$, 
as in~(\ref{eq:Dpm-unif}).
The corresponding results for $0\le\theta< O(P^{1/3})$ are presented
in section~\ref{subsec:density}. 
Again, the results presented in this section and their derivations do not depend on the step edge curvature.

\subsection{\em Step stiffness (section~$\ref{subsec:stiff}$)} \label{stepStiffness}
Let the adatom fluxes $f_\pm$ from an upper ($+$) and lower ($-$) 
terrace towards an edge be the input, independent parameters. For sufficiently
small angle $\theta$, the stiffness is found to be
\begin{equation}
\frac{\tbe}{k_B T}\sim 2\frac{l_{123}}{n_{123}}\,\frac{(\fp+\fm)_\theta}{\fp+\fm}\,\frac{1}{\theta}=
O\biggl(\frac{1}{\theta}\biggr)\qquad O(P^{1/3})<\theta\ll 1~,
\label{eq:stiff-smalltheta}
\end{equation}
\begin{equation}
\frac{\tbe}{k_B T}\sim P^{-2/3}\,\frac{4l_{123}}{n_{123} (\check C^k_0)^2+8m_{123}}=O(P^{-2/3})
\qquad 0\le\theta <O(P^{1/3})\ll 1~,
\label{eq:stiff-rg2}
\end{equation}
where $p_{123}$ ($p=l,\,m,\,n$) is defined in~(\ref{eq:sums-def}) 
and~\footnote{The superscripts 
in $\check C^k$, $C^\phi$, $C^k$ and elsewhere below indicate
the physical origin of these coefficients, and should not be confused with numerical exponents
or perturbation orders.}
\begin{equation}
\check C^k_0 =\biggl[\frac{2m_{123}}{n_{123}\,l_{123}}(\fp +\fm)\biggl]^{1/3}~.\label{eq:check-Ck}
\end{equation}
Matching the asymptotic  results~(\ref{eq:stiff-smalltheta}) 
and~(\ref{eq:stiff-rg2}) is discussed near the end of  section \ref{subsec:pertb}.

For $\theta=O(1)$ the formula for $\tbe$ becomes more complicated; we give it here
for completeness. Generally,
\begin{equation}
\frac{\tbe}{k_BT}=\frac{\phi^{(1)}}{\phi^{(0)}}~,\label{eq:stiff-intro}
\end{equation}
where $\phi^{(0)}$ and $\phi^{(1)}$ are expansion coefficients
for $\phi$ and depend on $f_\pm$ and their derivatives in $\theta$; 
cf.~(\ref{eq:phk-app}). These coefficients are obtained explicitly for $P\ll 1$.
In particular, for $O(P^{1/3})<\theta<\pi/4$,
\begin{equation}
\phi^{(0)}\sim C^\phi_0\,P~,\qquad
\phi^{(1)}\sim C^\phi_1\,P~,
\label{eq:ph01-asymp1}
\end{equation}
\begin{equation}
C^\phi_0=\frac{1}{2\sin\theta}\,\frac{(1+l_3\tan\theta)\fp+(1+l_2\tan\theta)\fm}
{Q({\bf l})}~,\label{eq:Cphi-def}
\end{equation}
\begin{equation}
C^\phi_1=\frac{(v^{(0)}_\theta+v^{(0)}\tan\theta)k^{(0)}_\theta\cos\theta+(w^{(0)}\tan\theta)_\theta}{2\sin\theta}\,
\frac{W^k\tan\theta+w^{(0)}+H^k}{H^k\,W^\phi}~,\label{eq:Cph1-def}
\end{equation}
\begin{equation}
v^{(0)}=f_++f_-~,\label{eq:v0-def}
\end{equation}
\begin{equation}
w^{(0)}\sim 2l_1 C^\phi_0+l_2\frac{s_xf_+ +2C^\phi_0}{1+l_2\tan\theta}+l_3\frac{s_xf_- +2C^\phi_0}{1+l_3\tan\theta}~,\qquad
s_x=(\cos\theta)^{-1}~,
\label{eq:w0-def}
\end{equation}
\begin{equation}
W^\phi= 2l_1+\frac{2l_2}{1+l_2\tan\theta}+\frac{2l_3}{1+l_3\tan\theta}~,
\label{eq:Wphi}
\end{equation}
\begin{equation}
W^k=-l_2\frac{l_2+\textstyle{\frac{n_2}{2}}\tan\theta}{(1+l_2\tan\theta)^2}(s_xf_+ +2C^\phi_0)
-l_3\frac{l_3+\textstyle{\frac{n_3}{2}}\tan\theta}{(1+l_3\tan\theta)^2}(s_x f_- +2C^\phi_0)~,
\label{eq:Wk}
\end{equation}
\begin{equation}
H^k=\frac{\tan\theta}{2}\biggl(2n_1C^\phi_0+n_2\,\frac{s_xf_+ +2C^\phi_0}{1+l_2\tan\theta}
+n_3\,\frac{s_xf_-+2C^\phi_0}{1+l_3\tan\theta}\biggr)~,
\label{eq:Hk}
\end{equation}
\begin{equation}
k\sim k^{(0)}\sim \tan\theta~,
\label{eq:k0-asymp1}
\end{equation}
where $n_j$ and $l_j$
are coordination numbers. Recall definition~(\ref{eq:Q-def}) for $Q({\bf l})$. 
It is worthwhile noting that there is no asymmetry in the step stiffness $\tilde \beta$, in contrast to the attachment-detachment coefficients. 
The reason for this difference is that $\tilde \beta$ depends only on the edge-atom density, as shown in (\ref{eq:stiff-intro}).\looseness=-1

For the alternative approach in which the adatom densities $\rho_\pm$ are specified rather than the fluxes $f_\pm$, 
the analysis of the step stiffness is presented in section \ref{subsec:extension}.
The corresponding result (\ref{eq:fpm-ss}) is not of the form (\ref{eq:f-noperm}) and (\ref{eq:GT}), however,  
since the coefficient $\be$ in (\ref{eq:ss-def}) is not proportional to $\rho_*$.

\section{The kinetic steady state }
\label{sec:steady}

We analyze the kinetic steady state for a straight step, including its dependence on the 
P\'eclet number $P$ in section \ref{sec:steadyPdependence} and  the ES effect
 and step permeabiity in sections \ref{subsec:flux} and \ref{subsec:density}. 

\subsection{Kinetic steady state and its dependence on $P$}
\label{sec:steadyPdependence}

In this section, we simplify the equations of motion for edge-atom and kink densities
by imposing the kinetic steady state ($\partial_t\equiv 0$) for straight steps ($\kappa\equiv 0$). 
We find closed-form solutions for
small P\'eclet number, $P\ll 1$, in two distinct ranges of $\theta$. For $\theta_c=O(P^{1/3})<\theta<\pi/4$, we show
that $\phi=\phi^{(0)}$ is given by~(\ref{eq:ph01-asymp1}), and $k=k^{(0)}$ is given by~(\ref{eq:k0-asymp1}),
or more precisely by
\begin{equation}
k^{(0)}\sim \tan\theta +C_0^k\,P,~\label{eq:k0-asymp1-corr}
\end{equation}
where
\begin{equation}
C_0^k=\frac{2C_0^\phi}{\tan\theta}\frac{2C_0^\phi Q({\bf m})\cos\theta+m_2(1+l_3\tan\theta)\fp +m_3(1+l_2\tan\theta)\fm}
{2C_0^\phi Q({\bf n})\cos\theta+n_2(1+l_3\tan\theta)\fp +n_3(1+l_2\tan\theta)\fm}~;
\label{eq:Ck-def}
\end{equation}
$Q({\bf p})$ and $C_0^\phi$ are defined by~(\ref{eq:Q-def}) and~(\ref{eq:Cphi-def}). 
Furthermore,
\begin{equation}
\phi^{(0)}\sim \check C^\phi_0\,P^{2/3}~,\quad k^{(0)}\sim \check C^k_0\,P^{1/3}\qquad 0\le\theta< \theta_c=O(P^{1/3})~,
\label{eq:phik-exp-2}
\end{equation}
where $\check C^k_0$ is defined by~(\ref{eq:check-Ck}),
\begin{equation}
\check C^\phi_0=\biggl(\frac{n_{123}}{4m_{123}}\biggr)^{1/3}\ \biggl(\frac{f_+ +f_-}{2l_{123}}\biggr)^{2/3}~,
\label{eq:check-Cphi}
\end{equation}
and $p_{123}$ ($p=l,\,m,\,n$) is given in~(\ref{eq:sums-def});
cf. equations~(4.27) and~(4.28) in~\cite{caflischli03}. In effect, we determine mesoscopic kinetic rates, including
the attachment-detachment and permeability coefficients in~(\ref{eq:ES-Dpm})--(\ref{eq:Apm-def}).

We proceed to describing the derivations. By $\partial_t|_\theta = 0$ and $\kappa=0$ in~(\ref{eq:phi-pde}), (\ref{eq:k-pde}) and~(\ref{eq:v-def}),
we have $f_+ + f_- = f_0\cos\theta$, $g=h$ and $v=f_0\cos\theta$.
Eliminate $\Dt\rho$ in terms of $\De\phi=2P^{-1}\phi$ using (\ref{eq:fpm-simp}).
Thus, we readily obtain~(\ref{eq:v0-def}) for $v^{(0)}:=v$, along with the following system of coupled algebraic equations:
\begin{eqnarray}
&&\bigl[m_1\phi^{(0)}-\textstyle{\frac{n_1}{4}}({k^{(0)}}^2-\tan^2\theta)\bigr]\,
\bigl[1+l_2k^{(0)}+m_2\phi^{(0)}+\textstyle{\frac{n_2}{4}}({k^{(0)}}^2-\tan^2\theta)\bigr]\nonumber\\
\mbox{}&&\times\bigl[1+l_3k^{(0)}+m_3\phi^{(0)}+\textstyle{\frac{n_3}{4}}({k^{(0)}}^2-\tan^2\theta)\bigr]2P^{-1}\phi^{(0)}
+\bigl[m_2\phi^{(0)}-\textstyle{\frac{n_2}{4}}({k^{(0)}}^2-\tan^2\theta)\bigr]\nonumber\\
\mbox{}&&\times\bigl[1+l_3 k^{(0)}+m_3\phi^{(0)}+\textstyle{\frac{n_3}{4}}({k^{(0)}}^2-\tan^2\theta)\bigr]\,
(s_x f_++2P^{-1}\phi^{(0)})\nonumber\\
&&+\bigl[m_3\phi^{(0)}-\textstyle{\frac{n_3}{4}}({k^{(0)}}^2-\tan^2\theta)\bigr]
\bigl[1+l_2k^{(0)}+m_2\phi^{(0)}+\textstyle{\frac{n_2}{4}}({k^{(0)}}^2-\tan^2\theta)\bigr]\nonumber\\
&&\hskip30pt \times(s_xf_-+2P^{-1}\phi^{(0)})=0~,
\label{eq:steady-phik-1}
\end{eqnarray}
\begin{eqnarray}
&&(l_1k^{(0)}+3m_1\phi^{(0)})\,[1+l_2 k^{(0)}+m_2\phi^{(0)}+\textstyle{\frac{n_2}{4}}({k^{(0)}}^2-\tan^2\theta)\bigr]\nonumber\\
&&\times \bigl[1+l_3 k^{(0)}+m_3\phi^{(0)}+\textstyle{\frac{n_3}{4}}({k^{(0)}}^2-\tan^2\theta)\bigr]2P^{-1}\phi^{(0)}+
(l_2k^{(0)}+3m_2\phi^{(0)})\nonumber\\
&&\times [1+l_3k^{(0)}+m_3\phi^{(0)}+\textstyle{\frac{n_3}{4}}({k^{(0)}}^2-\tan^2\theta)\bigr]
(s_xf_++2P^{-1}\phi^{(0)})+(l_3k^{(0)}+3m_3\phi^{(0)})\nonumber\\
&&\times[1+l_2k^{(0)}+m_2\phi^{(0)}+\textstyle{\frac{n_2}{4}}({k^{(0)}}^2-\tan^2\theta)\bigr](s_xf_-+2P^{-1}\phi_0)\nonumber\\
\mbox{}&&= (f_+ +f_-)\,s_x\bigl[1+l_2k_0+m_2\phi_0+\textstyle{\frac{n_2}{4}}({k^{(0)}}^2-\tan^2\theta)\bigr]\nonumber\\
\mbox{}&&\hskip25pt \times[1+l_3 k^{(0)}+m_3\phi^{(0)}+\textstyle{\frac{n_3}{4}}({k^{(0)}}^2-\tan^2\theta)\bigr]~.
\label{eq:steady-phik-2}
\end{eqnarray}
Once these equations are solved, the flux variables
$w=:w^{(0)}$, $g=:g^{(0)}$ and $h=:h^{(0)}$ are determined in terms of $f_\pm$ by the constitutive 
laws~(\ref{eq:w-def})--(\ref{eq:h-def}). The substitution of $\phi$ and $k$
into~(\ref{eq:fpm-simp}) provides a relation between
$f_\pm$ and $\rho_\pm$.

Next, we simplify and explicitly solve~(\ref{eq:steady-phik-1}) and~(\ref{eq:steady-phik-2})
by enforcing $P\ll 1$. The ensuing scaling of $\phi^{(0)}$ and $k^{(0)}$ with $P$ depends on the range of $\theta$.
We distinguish the cases $\theta_c(P)<\theta<\pi/4$ and $0\le \theta<\theta_c(P)$, where
$\theta_c$ is estimated below; we expect that $\theta_c\to 0$ as $P\to 0$. 
\vskip5pt

(i)\ $\theta=O(1)$. By seeking solutions that are regular at $P=0$, we observe that if $P=0$ then
$(\phi^{(0)}, k^{(0)})=(0, \tan\theta)$ solves~(\ref{eq:steady-phik-1}) and~(\ref{eq:steady-phik-2}).
Thus, the expansions
\begin{equation}
\phi^{(0)}\sim C^\phi_0\,P~,\qquad k^{(0)}\sim \tan\theta+C^k_0\,P~,\qquad C_0^{\phi,k}=O(1)~,
\label{eq:phik-exp}
\end{equation}
form a reasonable starting point. These expansions yield the simplified system
\begin{eqnarray}
\lefteqn{2C^\phi_0(2m_1C^\phi_0-n_1C^k_0\tan\theta)
(1+l_2\tan\theta)(1+l_3\tan\theta)}\nonumber\\
&&+(2m_2C^\phi_0-n_2 C^k_0\tan\theta)(1+l_3\tan\theta)(s_xf_++2C^\phi_0)\nonumber\\
&&+(2m_3C^\phi_0-n_3C^k_0\tan\theta)(1+l_2\tan\theta)(s_x f_-+2C^\phi_0)=0~,
\label{eq:Cphik-1}
\end{eqnarray}
\begin{eqnarray}
\lefteqn{2C^\phi_0\, l_1\tan\theta(1+l_2\tan\theta)(1+l_3\tan\theta)}\nonumber\\
&&+l_2\tan\theta(1+l_3\tan\theta)(s_x f_+ +2C^\phi_0)+l_3\tan\theta(1+l_2\tan\theta)(s_x f_-+2C^\phi_0)\nonumber\\
\mbox{}&&=s_x(f_+ +f_-)(1+l_2\tan\theta)(1+l_3\tan\theta)~.\label{eq:Cphik-2}
\end{eqnarray}The solution of this system leads to~(\ref{eq:Cphi-def}) and~(\ref{eq:Ck-def}).

We now sketch an order-of-magnitude estimate for $\theta_c$, the lower bound
for $\theta$ in the present range of interest. By~(\ref{eq:Ck-def}), $C^k_0=O(1/\theta^2)$ for 
$\theta_c<\theta\ll 1$. 
Hence, expansion~(\ref{eq:phik-exp})
for the kink density $k^{(0)}$ breaks down when its leading-order term, $\tan\theta$, is comparable
to the correction term, $C^k_0 P$: $\theta_c=O(P/\theta_c^2)$ by which $\theta_c=O(P^{1/3})$.
Thus,~(\ref{eq:phik-exp})--(\ref{eq:Cphik-2}) hold if $O(P^{1/3})<\theta<\pi/4$.
A more accurate estimate of the lower bound requires the detailed solution of~(\ref{eq:steady-phik-1}) and
(\ref{eq:steady-phik-2}) for $\theta=O(P^{1/3})$, and will not be pursued here.
\vskip5pt

(ii)\ $0\le\theta< O(P^{1/3})$. For all practical purposes we set $\theta=0$ in~(\ref{eq:steady-phik-1})
and~(\ref{eq:steady-phik-2}). We enforce the expansions
\begin{equation}
\phi^{(0)}\sim \check C^\phi_0\,P^{\nu}~,\qquad k^{(0)}\sim \check C^k_0\,P^{\sigma}\qquad 
\check C_0^{\phi,k}=O(1)\quad \mbox{as}\ P\to 0~,
\label{eq:phik-exp-2-mnu}
\end{equation}
and find the exponents $\nu$ and $\sigma$ by {\it reductio ad absurdum}. The only values
consistent with~(\ref{eq:steady-phik-1}) and~(\ref{eq:steady-phik-2}) readily turn out to be
\begin{equation}
\nu=2/3~,\qquad \sigma=1/3~.
\label{eq:nu-sigma}
\end{equation}
These values are in agreement with the analysis in~\cite{caflischli03}.
By dominant-balance arguments, the coefficients $\check C^\phi_0$ and $\check C^k_0$ satisfy
\begin{equation}
4m_{123}\check C^\phi_0=n_{123}\,(\check C^k_0)^2~,\qquad 2l_{123}\,\check C^\phi_0\,\check C^k_0=f_+ +f_-~,
\label{eq:check-Cphik}
\end{equation}
by which we readily obtain~(\ref{eq:check-Ck}) and~(\ref{eq:check-Cphi}).
Note that the zeroth-order kink velocity becomes
\begin{equation}
w=w^{(0)}\sim 2l_{123}P^{-1/3} \check C^\phi_0=O(P^{-1/3})~.
\label{eq:w0-rg2}
\end{equation}
\vskip5pt

(iii)\ {\it Consistency of asymptotics for $\theta=O(P^{1/3})$}. As a check 
on the consistency of our asymptotics
and the estimate of $\theta_c$, we study the limits of~(\ref{eq:phik-exp-2}) and~(\ref{eq:phik-exp})
in the transition region, as $\theta\to O(P^{1/3})$. It is expected that the two sets
of formulas for $\phi$ and $k$, in $O(P^{1/3})<\theta<\pi/4$ and $0\le \theta<O(P^{1/3})$, should furnish the same order of
magnitudes.

Indeed, by letting $\theta\to O(P^{1/3})\ll 1$ in~(\ref{eq:phik-exp}) for $\phi$ 
we find $\phi=O(P/\theta)\to O(P^{2/3})$,
in agreement with~(\ref{eq:phik-exp-2}) for $\nu=2/3$. Similarly, setting $\theta=O(P^{1/3})$
in~(\ref{eq:phik-exp}) for $k$ yields $k=O(\theta)\to O(P^{1/3})$, which is consistent with~(\ref{eq:phik-exp-2})
for $\sigma=1/3$. In section~\ref{subsec:stiff} we show that such a ``matching'' is not always achieved
for the first-order corrections $\phi^{(1)}$ and $k^{(1)}$, since the corresponding
asymptotic formulas involve derivatives in $\theta$. A sufficient condition on the $\theta$-behavior
of the fluxes $f_\pm$ is sought in the latter case.\looseness=-1

In the following, we use the kinetic steady state in the mean-field law~(\ref{eq:fpm}) 
to derive mesoscopic kinetic rates as functions of $\theta$ by comparison
to the BCF-type equation~(\ref{eq:f-bcf}). We adopt two approaches.
In the first approach, $f_\pm$ are used as external, input parameters; the effective kinetic coefficients
are thus allowed to depend on $f_\pm$. In the second approach, the densities $\rho_\pm$
are the primary variables instead. 

\subsection{Flux-driven kinetics approach: ES effect}
\label{subsec:flux}

In this subsection we treat the fluxes $f_\pm$ as given, input parameters.
Accordingly, we derive~(\ref{eq:ES-Dpm}) and~(\ref{eq:Dpm-unif}), i.e. the attachment-detachment 
rates $\Da^\pm$ for adatoms. In addition, we show that the reference densities $\rho_*^\pm$ entering~(\ref{eq:GT}) are
\begin{equation}
\rho_*^\pm\sim \frac{2C^\phi_0}{\Dt(1+l_{j_\pm}\tan\theta)}\qquad j_+=2,\,j_-=3~,\quad O(P^{1/3})<\theta<\pi/4~,
\label{eq:rho-star-range1}
\end{equation}
\begin{equation}
\rho_*^\pm\sim \frac{2\check C^\phi_0}{\Dt}P^{-1/3}\qquad 0\le\theta< O(P^{1/3})~,
\label{eq:rho-star-range2}
\end{equation}
where $C^\phi_0$ and $\check C^\phi_0$ are defined by~(\ref{eq:Cphi-def}) and~(\ref{eq:check-Cphi}).

The relevant derivations follow. The substitution of~(\ref{eq:phik-exp}) into~(\ref{eq:fpm-simp}) yields
\begin{equation}
f_\pm \sim \biggl[1+l_{j_\pm}k^{(0)}+m_{j_\pm}\phi^{(0)}+\frac{n_{j_\pm}}{4}({k^{(0)}}^2-\tan^2\theta)\biggr]\Dt\rho_\pm\cos\theta-
2P^{-1}\phi^{(0)}\cos\theta~.\label{eq:f+0}
\end{equation}
Here, we view the linear-in-$\rho_\pm$ term of~(\ref{eq:f+0}) as the only physical contribution
of the adatom densities to the mass flux towards an edge. 
Consequently, by comparison to~(\ref{eq:f-bcf}),
the coefficient of this term must be identified with $\Da^\pm$. Thus,
we extract formulas~(\ref{eq:ES-Dpm}). In addition, we
obtain $D_p^\pm\equiv 0$; so, step permeability is {\it not}
manifested in this context. The reference density of~(\ref{eq:GT}) is\looseness=-1
\begin{equation}
\rho_*^\pm = \frac{2P^{-1}\phi^{(0)}\cos\theta}{\Da^\pm}~,
\label{eq:rho-star}
\end{equation}which is in principle {\it different} for an up- and down-step edge.

Note that $\Da^\pm$ and $\rho_*^\pm$ depend on $f_\pm$ within this approach.
Further, the ratio of $\Da^+$ and $\Da^-$ depends on the
values of $l_j$, $m_j$ and $n_j$. For suitable coordination numbers,
it is possible to have $\Da^+ > \Da^-$, i.e. a negative (vs. positive) ES 
effect~\cite{ehrlichhudda66,schwoebelshipsey66}, which can lead to instabilities in the step motion.
Next, we derive simplified, explicit formulas for $\Da^\pm$ and $\rho_*^\pm$ when $P\ll 1$.
\vskip5pt

(i)\ $O(P^{1/3})<\theta<\pi/4$. By substitution of~(\ref{eq:phik-exp}) with~(\ref{eq:Cphi-def})
and~(\ref{eq:Ck-def}) into~(\ref{eq:ES-Dpm}) we have
\begin{eqnarray}
\Da^+&\sim& \Dt\bigl[1+l_2\tan\theta+\bigl(l_2+\textstyle{\frac{1}{2}}n_2\tan\theta\bigr)C^k_0 P+
m_2C^\phi_0 P\bigr]\cos\theta~,\nonumber\\
\Da^-&\sim& \Dt\bigl[1+l_3\tan\theta+\bigl(l_3+\textstyle{\frac{1}{2}}n_3\tan\theta\bigr)C^k_0 P+
m_3C^\phi_0 P\bigr]\cos\theta~,
\label{eq:DpmI-range1}
\end{eqnarray}which reduce to~(\ref{eq:Dpm-unif}) as $P\to 0$.

In the same vein, by~(\ref{eq:rho-star}) the reference densities $\rho_*^\pm$ are
\begin{eqnarray}
\rho_*^+&\sim&\frac{2C^\phi_0}
{\Dt\bigl[1+l_2\tan\theta+\bigl(l_2+\textstyle{\frac{1}{2}}n_2\tan\theta\bigr)C^k_0 P+
m_2C^\phi_0 P\bigr]}~,\nonumber\\
\rho_*^-&\sim&\frac{2C^\phi_0}
{\Dt\bigl[1+l_3\tan\theta+\bigl(l_3+\textstyle{\frac{1}{2}}n_3\tan\theta\bigr)C^k_0 P+
m_3C^\phi_0 P\bigr]}~,
\label{eq:rho-starI-rangeI}
\end{eqnarray}
which readily yield~(\ref{eq:rho-star-range1}).
\vskip5pt

(ii)\ $0\le\theta< O(P^{1/3})$. In this case, we resort to~(\ref{eq:phik-exp-2}). 
Equation~(\ref{eq:ES-Dpm}) for the kinetic rates furnishes
\begin{eqnarray}
\Da^+&\sim& \Dt\bigl\{ 1+l_2{\check C}^k_0\,P^{1/3}+\bigl[ m_2{\check C}^\phi_0+\textstyle{\frac{1}{4}}n_2
({\check C}^k_0{})^2\bigr]P^{2/3}\bigr\}
\cos\theta\nonumber\\
&\sim& \Dt(1+l_2{\check C}^k_0\,P^{1/3})\cos\theta~,\nonumber\\
\Da^-&\sim& \Dt\bigl\{ 1+l_3{\check C}^k_0\,P^{1/3}+\bigl[m_3{\check C}^\phi_0+\textstyle{\frac{1}{4}}n_2
({\check C}^k_0{})^2\bigr]P^{2/3}\bigr\}
\cos\theta\nonumber\\
&\sim& \Dt(1+l_3{\check C}^k_0\,P^{1/3})\cos\theta~.\label{eq:Dpm-range2}
\end{eqnarray}
To leading order in $P$, these formulas connect smoothly with~(\ref{eq:DpmI-range1})
and, thus, justify~(\ref{eq:Dpm-unif}) for $0\le\theta<\pi/4$.
Furthermore, $\rho_*^\pm$ are given by
\begin{eqnarray}
\rho_*^+ &\sim& \frac{2\check C^\phi_0\,P^{-1/3}}{\Dt(1+l_2\check C^k_0\,P^{1/3})}
\sim \frac{2\check C^\phi_0}{\Dt}P^{-1/3}(1-l_2\check C^k_0\,P^{1/3})~,\nonumber\\
\rho_*^- &\sim& \frac{2\check C^\phi_0\,P^{-1/3}}{\Dt(1+l_3\check C^k_0\,P^{1/3})}
\sim \frac{2\check C^\phi_0}{\Dt}P^{-1/3}(1-l_3\check C^k_0\,P^{1/3})~,
\label{eq:rho-starI-range2}
\end{eqnarray}which reduce to~(\ref{eq:rho-star-range2}). Notably, $\rho_*^\pm$ depend 
on the fluxes, $f_\pm$, through $\check C_0^\phi$.

A few remarks are in order. First, by~(\ref{eq:Dpm-unif}) the ES
effect is present for $O(P^{1/3})<\theta<\pi/4$ only 
if $l_2\neq l_3$. Accordingly, our formalism provides explicitly an analytical relation between
the number of transition paths for atomistic processes and the mesoscopic kinetic rates.
Second, formulas~(\ref{eq:Dpm-range2}) show that for $P\ll 1$ the nonzero ES barrier is a corrective, 
$O(P^{1/3})$ effect
for sufficiently small $\theta$, even when $l_2\neq l_3$. 

\subsection{Density-driven approach: Step permeability and ES effect}
\label{subsec:density}
In this subsection we show that the treatment of the densities $\rho_\pm$ as independent, external parameters in
the kinetic law~(\ref{eq:fpm-simp}) leads to coexistence of the ES effect and step permeability.
In particular, the permeability and attachment-detachment rates
are provided by~(\ref{eq:perm-Dpm})--(\ref{eq:Apm-def}).

To derive~(\ref{eq:perm-Dpm}) and~(\ref{eq:Dpm-ES-dens}),
we solve~(\ref{eq:f+0}) for $f_\pm$, which are viewed as {\it dependent} variables,
taking into account that $\phi^{(0)}$ and $k^{(0)}$ depend on $f_\pm$. To simplify the algebra
while keeping the essential physics intact, we restrict attention to $O(P^{1/3})<\theta< \pi/4$.

First, in view of~(\ref{eq:phik-exp}) we further simplify relation~(\ref{eq:fpm-simp}). By 
\begin{equation}
\phi^{(0)}=\textstyle{\frac{1}{2}}(A_+ f_+ + A_- f_-) P~,\label{eq:phi0-alt}
\end{equation}
where $A_\pm$ are defined by~(\ref{eq:Apm-def}), the adatom fluxes at the step edge reduce to
\begin{equation}
f_\pm\sim (1+l_{j_\pm}\tan\theta)\Dt\rho_\pm\cos\theta-(A_+ f_+ + A_- f_-)\cos\theta~,\qquad P\ll 1~.
\label{eq:fpm-alt}
\end{equation}

Second, we invert~(\ref{eq:fpm-alt}) to obtain $f_\pm$ in terms of $\rho_\pm$.
Equation (\ref{eq:fpm-alt}) reads
\begin{eqnarray}
(1+A_+\cos\theta)f_+ +\cos\theta\,A_- f_-&=&(1+l_2\tan\theta)\Dt\rho_+\cos\theta~,\nonumber\\
\cos\theta\,A_+ f_+ + (1+\cos\theta\,A_-)f_-&=& (1+l_3\tan\theta)\Dt\rho_-\cos\theta~.
\label{eq:A+-alt}
\end{eqnarray}
The inversion of this system yields
\begin{eqnarray}
&&f_+ = \biggl[\frac{(1+l_2\tan\theta)(1+\cos\theta\,A_-)}{1+(A_+ +A_-)\cos\theta}\Dt\rho_+
-\frac{A_- (1+l_3\tan\theta)\cos\theta}{1+(A_+ + A_-)\cos\theta}\Dt\rho_-\biggr]\cos\theta~,\nonumber\\
&&f_- = \biggl[-\frac{(1+l_2\tan\theta)A_+\cos\theta}{1+(A_+ +A_-)\cos\theta}\Dt\rho_+
+\frac{(1+A_+ \cos\theta)(1+l_3\tan\theta)}{1+(A_+ + A_-)\cos\theta}\Dt\rho_-\biggr] \cos\theta~.\nonumber
\end{eqnarray}
These relations have the form of the kinetic law~(\ref{eq:f-bcf}); by comparison, the rates 
$D_p^\pm$ are given by~(\ref{eq:perm-Dpm}), while
\begin{equation}
\rho_0^\pm \equiv 0\Rightarrow \rho_*^\pm\equiv 0~.
\label{eq:rho-star-perm}
\end{equation}This value is expected since the system
is homogeneous in this setting, i.e. $f_\pm=0$ only if $\rho_\pm=0$. The reference
density $\rho_0$ becomes nonzero (but small in an appropriate sense) if we allow in the formulation
nonzero values for $D_B$ and $D_K$, i.e. nonzero diffusion coefficients
for an atom to hop from a kink and a straight edge. The study of these effects lies beyond our
present scope. Equation (\ref{eq:rho-star-perm}) challenges the definition 
of the step stiffness; see section~\ref{subsec:extension}. 

Equations (\ref{eq:A+-alt}) also predict an ES effect. Indeed, by
recourse to~(\ref{eq:f-bcf}), the related attachment-detachment rates are
\begin{equation}
\Da^\pm =\frac{(1+l_{j_\pm}\tan\theta)(1+A_\mp\cos\theta)}{1+(A_+ +A_-)\cos\theta}\Dt\cos\theta-D_p^\pm~,
\label{eq:Dpm-prim}
\end{equation}
which readily yields~(\ref{eq:Dpm-ES-dens}) by use of~(\ref{eq:perm-Dpm}).

The behavior of the fluxes $f_\pm$ as functions of $\rho_\pm$ is 
dramatically different for $0\le\theta < O(P^{1/3})$. 
Indeed, by~(\ref{eq:phik-exp-2})--(\ref{eq:check-Cphi}) 
the density $\phi^{(0)}$ is a nonlinear algebraic 
function of $f_+ + f_-$ in this case. Thus, 
the mean-field constitutive equations in principle {\it cannot} reduce to
kinetic laws that are linear in $\rho_\pm$. 
This approach does not lead to standard BCF-type conditions at a high-symmetry
step edge orientation. The implications of this behavior warrant further studies.

In the following
analysis for the stiffness we emphasize the flux-driven approach.

\section{Perturbation theory and step stiffness}
\label{sec:GT-stiff}
In this section we consider slightly curved step edges, and apply perturbation theory to
find approximately the edge-atom and kink densities, $\phi$ and $k$, from the kinetic
model of section~\ref{subsec:noneq-kin}. On the basis of the linear kinetic law~(\ref{eq:f-bcf}) 
along with~(\ref{eq:GT}) for $\rho_0$,
we calculate the step stiffness, $\tbe$, as a function of the orientation angle, $\theta$; 
see formulas~(\ref{eq:stiff-smalltheta})--(\ref{eq:k0-asymp1}). 
The underlying perturbation scheme for the densities is outlined in appendix~\ref{app:lin-pert}.

The starting point is expansion~(\ref{eq:phk-app}), which we assume to be valid for 
$0\le\theta < \pi/4$ and view as a Taylor series.
The functions $\phi^{(0)}$ and $k^{(0)}$ correspond to
the kinetic steady state of section~\ref{sec:steadyPdependence}. The first-order coefficients $\phi^{(1)}$ and $k^{(1)}$ are
locally bounded and are evaluated
below. Only the coefficient $\phi^{(1)}$ is needed for the calculation of the step stiffness, $\tbe$,
by~(\ref{eq:fpm-simp});
for completeness, we also derive $k^{(1)}$.

The relation of $\tbe$ to $\phi^{(0)}$ and $\phi^{(1)}$ is provided by the following argument. By substitution
of~(\ref{eq:phk-app}) into~(\ref{eq:fpm-simp}) and treatment of $f_\pm$ as given external parameters
(in the spirit of section~\ref{subsec:flux}), we obtain
\begin{equation}
\frac{f_\pm}{\cos\theta}=[1+l_{j_\pm}k^{(0)}+m_{j_\pm}\phi^{(0)}+\textstyle{\frac{n_{j_\pm}}{4}}
({k^{(0)}}^2-\tan^2\theta)]\Dt\rho_\pm -\De\phi^{(0)}-\kappa \De\phi^{(1)}~,
\label{eq:fpm-kappa}
\end{equation}where $j_+=2$ and $j_-=3$. By comparison of~(\ref{eq:fpm-kappa})
to~(\ref{eq:GT}) and~(\ref{eq:f-bcf}), we have (using $2P^{-1}=D_E$)
\begin{equation}
\Da^\pm\rho_*^\pm\,\frac{\tbe}{k_B T}=2P^{-1}\,\phi^{(1)}\,\cos\theta~,\qquad \label{eq:lten}
\end{equation}
by which we assert~(\ref{eq:stiff-intro}) in view of~(\ref{eq:rho-star}). Our task is to 
calculate $\phi^{(1)}$ in terms of $\theta$ and $P$ when $P\ll 1$.

\subsection{Linear perturbations}
\label{subsec:pertb}
In this subsection we derive formula~(\ref{eq:ph01-asymp1}) for $\phi^{(1)}$ along with 
(\ref{eq:Cph1-def}) and (\ref{eq:Wphi})--(\ref{eq:Hk}) 
when $O(P^{1/3})<\theta <\pi/4$. In addition, we show that in this regime
\begin{equation}
k^{(1)}\sim -\frac{(v^{(0)}_\theta+v^{(0)}\tan\theta)
(\cos\theta)^{-1}+(w^{(0)}\tan\theta)_\theta}{2H^k\cos\theta}~.
\label{eq:k1-app1}
\end{equation}For $0\le\theta < O(P^{1/3})$, $\phi^{(1)}$ and $k^{(1)}$ are
\begin{equation}
\phi^{(1)}\sim 4l_{123}\,\frac{\check C^\phi_0}{n_{123} (\check C^k_0)^2+8m_{123}}=O(1)~,
\label{eq:phi1-rg2}
\end{equation}
\begin{equation}
k^{(1)}\sim -4P^{-1/3}l_{123}\,\frac{\check C^k_0}{n_{123}(\check C^k_0)^2+8m_{123}\check C^\phi_0}
=O(P^{-1/3})~.
\label{eq:k1-rg2}
\end{equation}
Recall that $\check C^k_0$ and $\check C^\phi_0$ are defined by~(\ref{eq:check-Ck})
and~(\ref{eq:check-Cphi}). 
Furthermore, we demonstrate that $|\kappa|$ should be bounded by $P$ for the perturbation theory to hold;
see~(\ref{eq:restr-kappa}).

We proceed to carry out the derivations.
Following appendix~\ref{app:lin-pert}, we formulate a $2\times 2$ system of linear perturbations
for $\phi$ and $k$. First, we linearize the algebraic, constitutive laws~(\ref{eq:w-def})--(\ref{eq:h-def}).
Expansions~(\ref{eq:phk-app}) induce the approximations
$w(\phi,k)\sim w(\phi^{(0)},k^{(0)})+\kappa\, w^{(1)}$, $g(\phi,k)\sim g(\phi^{(0)},k^{(0)})+\kappa\, g^{(1)}$ 
and $h(\phi,k)\sim h(\phi^{(0)},k^{(0)})+\kappa\, h^{(1)}$, where
\begin{eqnarray}
w^{(1)}&=&\phi^{(1)}\,w_\phi+k^{(1)}\,w_k\qquad [w_\phi:=\partial_\phi w(\phi,k)]~,\nonumber\\
g^{(1)}&=&\phi^{(1)}\,g_\phi+k^{(1)}\,g_k~,\quad h^{(1)}=\phi^{(1)}\,h_\phi+k^{(1)}\,h_k~.
\label{eq:wgh1}
\end{eqnarray}In addition, $v^{(0)}=f_+ +f_-$ and $g^{(0)}:=g(\phi^{(0)},k^{(0)})=h(\phi^{(0)},k^{(0)})=:h^{(0)}$.
Second, we replace the above expansions in the equations of motion~(\ref{eq:phi-pde-th}) and~(\ref{eq:k-pde-th})
and the constitutive law~(\ref{eq:f0}). Hence, we find the system
\begin{eqnarray}
\lefteqn{(w_\phi k^{(0)}+2g_\phi+h_\phi)\phi^{(1)}+(w_k k^{(0)}+w^{(0)}+2g_k+h_k)k^{(1)}=-(v^{(0)}_\theta+v^{(0)}\tan\theta)
\phi^{(0)}_\theta,}\nonumber\\
&&\mbox{}\hskip10pt 2(g_\phi-h_\phi)\phi^{(1)}+2(g_k-h_k)k^{(1)}=(v^{(0)}_\theta+v^{(0)}\tan\theta)k^{(0)}_\theta+s_x
(w^{(0)}\tan\theta)_\theta~,
\label{eq:phi-k-1}
\end{eqnarray}
where $w^{(0)}:=w(\phi^{(0)}, k^{(0)})$ and $s_x=1/\cos\theta$. This system has solution
\begin{equation}
\phi^{(1)}=\frac{\mathcal D^\phi}{\mathcal D}~,\qquad k^{(1)}=\frac{\mathcal D^k}{\mathcal D}~,
\label{eq:phi-k-1-sol}
\end{equation}
where
\begin{equation}
\mathcal D=\left|\begin{array}{ll}w_\phi k^{(0)}+(2g+h)_\phi\ &\ w_k k^{(0)}+w^{(0)}+(2g+h)_k\\
                                  2(g-h)_\phi\ &\ 2(g-h)_k\end{array}\right|~,\label{eq:D-determ}
\end{equation}
\begin{equation}
\mathcal D^\phi=\left|\begin{array}{ll}-(v^{(0)}_\theta+v^{(0)}\tan\theta)\phi^{(0)}_\theta & w_k k^{(0)}+w^{(0)}+2(g+h)_k\\
                                        (v^{(0)}_\theta+v^{(0)}\tan\theta)k^{(0)}_\theta+s_x(w^{(0)}\tan\theta)_\theta\
 & 2(g-h)_k\end{array}\right|,\label{eq:Dphi-determ}
\end{equation}
\begin{equation}
\mathcal D^k=\left|\begin{array}{ll}w_\phi k^{(0)} +(2g+h)_\phi & -(v^{(0)}_\theta+v^{(0)}\tan\theta)\phi^{(0)}_\theta\\
                                    2(g-h)_\phi\ &\ (v^{(0)}_\theta+v^{(0)}\tan\theta)k^{(0)}_\theta+s_x
(w^{(0)}\tan\theta)_\theta\end{array}\right|~.
\label{eq:Dk-determ}
\end{equation}Note that $\phi^{(1)}$ and $k^{(1)}$ depend on the $\theta$-derivatives
of the zeroth-order (kinetic steady-state) solutions. 

By~(\ref{eq:w-def})--(\ref{eq:h-def}), 
we calculate the $\phi$- and $k$-derivatives of $w$, $g$ and $h$:
\begin{eqnarray}
w_\phi&=&2l_1 P^{-1}+\sum_{q=+,-}l_{j_q}\biggl\{\frac{2P^{-1}}{1+l_{j_q} k^{(0)}+m_{j_q} \phi^{(0)}+
\frac{n_{j_q}}{4}({k^{(0)}}^2-\tan^2\theta)}\nonumber\\
&&\mbox{}-m_{j_q}\,\frac{(\cos\theta)^{-1}\,f_q+2P^{-1}\phi^{(0)}}
{[1+l_{j_q}k^{(0)}+m_{j_q}\phi^{(0)}+
\frac{n_{j_q}}{4}({k^{(0)}}^2-\tan^2\theta)]^2}\biggr\}~,\quad j_+ =2,\,j_-=3~,
\label{eq:wphi-deriv}
\end{eqnarray}
\begin{equation}
w_k=-\sum_{q=+,-}l_{j_q}\bigl(l_{j_q}+{\textstyle\frac{n_{j_q}}{2}}k^{(0)}\bigr)\,
\frac{(\cos\theta)^{-1}\,f_q +2P^{-1}\phi^{(0)}}{[1+l_{j_q}k^{(0)}+m_{j_q}\phi^{(0)}+
\frac{n_{j_q}}{4}({k^{(0)}}^2-\tan^2\theta)]^2}~,
\label{eq:wk-deriv}
\end{equation}
\begin{eqnarray}
g_\phi&=& 4m_1 P^{-1}\phi^{(0)}+\sum_{q=+,-}m_{j_q}\biggl\{\frac{(\cos\theta)^{-1}\,f_q +4P^{-1}\phi^{(0)}}{1+l_{j_q}k^{(0)}+m_{j_q}
\phi^{(0)}+\frac{n_{j_q}}{4}({k^{(0)}}^2-\tan^2\theta)}\nonumber\\
\mbox{} && -m_{j_q}\phi^{(0)}\frac{(\cos\theta)^{-1}\,f_q +2P^{-1}\phi^{(0)}}{[1+l_{j_q}k^{(0)}+m_{j_q}\phi^{(0)}+
\frac{n_{j_q}}{4}({k^{(0)}}^2-\tan^2\theta)]^2}\biggr\}~,\label{eq:gphi-deriv}
\end{eqnarray}
\begin{equation}
g_k = -\phi^{(0)}\sum_{q=+,-}m_{j_q}\frac{\bigl(l_{j_q}+{\textstyle\frac{n_{j_q}}{2}}k^{(0)}\bigr)[(\cos\theta)^{-1}\,f_q+2P^{-1}\phi^{(0)}]}
{[1+l_{j_q}k^{(0)}+m_{j_q}\phi^{(0)}+\frac{n_{j_q}}{4}({k^{(0)}}^2-\tan^2\theta)]^2}~,\label{eq:gk-deriv}
\end{equation}
\begin{eqnarray}
h_\phi&=& {\textstyle\frac{1}{4}}({k^{(0)}}^2-\tan^2\theta)\biggl\{2n_1 P^{-1}
+\sum_{q=+,-}\biggl[\frac{2P^{-1}n_{j_q}}{1+l_{j_q}k^{(0)}+m_{j_q}\phi^{(0)}+\frac{n_{j_q}}{4}({k^{(0)}}^2-\tan^2\theta)}\nonumber\\
&&\hskip20pt -m_{j_q}n_{j_q}\,\frac{(\cos\theta)^{-1}\,f_q+2P^{-1}\phi^{(0)}}{[1+l_{j_q}k^{(0)}+m_{j_q}\phi^{(0)}+\frac{n_{j_q}}{4}
({k^{(0)}}^2-\tan^2\theta)]^2}\biggr]\biggr\}~,
\label{eq:hphi-deriv}
\end{eqnarray}
\begin{eqnarray}
h_k&=& \frac{k^{(0)}}{2}\biggl[2n_1 P^{-1}\phi^{(0)}+\sum_{q=+,-}n_{j_q}\frac{(\cos\theta)^{-1}f_q+2P^{-1}\phi^{(0)}}
{1+l_{j_q}k^{(0)}+m_{j_q}\phi^{(0)}+\frac{n_{j_q}}{4}({k^{(0)}}^2-\tan^2\theta)}\biggr]\nonumber\\
&&\hskip20pt -{\textstyle\frac{1}{4}}({k^{(0)}}^2-\tan^2\theta)\sum_{q=+,-}n_{j_q}\frac{(l_{j_q}+
\frac{n_{j_q}}{2}k^{(0)})\,
[(\cos\theta)^{-1}f_q+2P^{-1}\phi^{(0)}]}{[1+l_{j_q}k^{(0)}+m_{j_q}\phi^{(0)}+\frac{n_{j_q}}{4}({k^{(0)}}^2-\tan^2\theta)]^2}~.
\label{eq:hk-deriv}
\end{eqnarray}

Equations~(\ref{eq:phi-k-1-sol})--(\ref{eq:hk-deriv}) are simplified under
the condition $P\ll 1$, which we apply next. We distinguish two ranges for the angle $\theta$.
\vskip5pt

(i)\ $O(P^{1/3})<\theta< \pi/4$. We proceed to show~(\ref{eq:ph01-asymp1}) 
and~(\ref{eq:Cph1-def}) for $\phi^{(1)}$. By using~(\ref{eq:phik-exp}) 
with~(\ref{eq:Cphi-def}) and~(\ref{eq:Ck-def}),
we replace $\phi^{(0)}$ and $k^{(0)}$ by their expansions in $P$. Thus, the derivatives
of $w$, $g$ and $h$ are simplified to
\begin{equation}
w_\phi \sim P^{-1}\biggl(2l_1+\frac{2l_2}{1+l_2\tan\theta}+\frac{2l_3}{1+l_3\tan\theta}\biggr)=:P^{-1}\, W^\phi=O(P^{-1}
)~,
\label{eq:wphi-app1}
\end{equation}
\begin{equation}
w_k\sim -\sum_{q=+,-}l_{j_q}\bigl(l_{j_q}+{\textstyle\frac{n_{j_q}}{2}}\tan\theta\bigr)
\frac{(\cos\theta)^{-1}f_q+2C^\phi_0}{(1+l_{j_q}\tan\theta)^2}=O(1)~,
\label{eq:wk-app1}
\end{equation}
\begin{equation}
g_\phi\sim 4m_1 C^\phi_0+\sum_{q=+,-}m_{j_q}\,\frac{(\cos\theta)^{-1}f_q+4 C^\phi_0}{1+l_{j_q}\tan\theta}=O(1)~,
\label{eq:gphi-app1}
\end{equation}
\begin{equation}
g_k\sim -P C^\phi_0\sum_{q=+,-}m_{j_q}\bigl(l_{j_q}+{\textstyle\frac{n_{j_q}}{2}}\tan\theta\bigr)
\frac{(\cos\theta)^{-1}f_q+2C^\phi_0}{(1+l_{j_q}\tan\theta)^2}=O(P)~,
\label{eq:gk-app1}
\end{equation}
\begin{equation}
h_\phi \sim C^k_0\,\tan\theta\biggl(n_1+\frac{n_2}{1+l_2\tan\theta}+\frac{n_3}{1+l_3\tan\theta}\biggr)=O(1)~,
\label{eq:hphi-app1}
\end{equation}
\begin{equation}
h_k \sim \frac{\tan\theta}{2}\biggl( 2n_1 C^\phi_0+\sum_{q=+,-}n_{j_q}
\frac{(\cos\theta)^{-1}f_q+2C^\phi_0}{1+l_{j_q}\tan\theta}\biggr)=O(1)~.
\label{eq:hk-app1}
\end{equation}

It follows that the determinants of~(\ref{eq:D-determ})--(\ref{eq:Dk-determ}) are
\begin{equation}
\mathcal D\sim -2P^{-1}\,h_kW^\phi\tan\theta~,
\label{eq:D-determ-app1}
\end{equation}
\begin{equation}
\mathcal D^\phi\sim -(w_k\tan\theta+w^{(0)}+h_k)[(v^{(0)}_\theta+v^{(0)}\tan\theta)k^{(0)}_\theta+(\cos\theta)^{-1}(w^{(0)}\tan\theta)_\theta]~,
\label{eq:Dphi-determ-app1}
\end{equation}
\begin{equation}
\mathcal D^k\sim P^{-1}W^\phi\frac{\tan\theta}{\cos\theta}[(v^{(0)}_\theta+v^{(0)}\tan\theta)
(\cos\theta)^{-1}+(w^{(0)}\tan\theta)_\theta]~.
\label{eq:Dk-determ-app1}
\end{equation}
Hence, in view of~(\ref{eq:phi-k-1-sol}), the coefficient $\phi^{(1)}$ is given by~(\ref{eq:ph01-asymp1})
with~(\ref{eq:Cph1-def}) and~(\ref{eq:v0-def})--(\ref{eq:k0-asymp1}) under
the replacements $H^k:=h_k$, $W^\phi:= Pw_\phi$ and $W^k:=w_k$. By~(\ref{eq:phi-k-1-sol}), the
corresponding coefficient $k^{(1)}$ is given by~(\ref{eq:k1-app1}).
\vskip5pt

(ii)\ $0\le\theta < O(P^{1/3})$. We now calculate the first-order corrections $\phi^{(1)}$
and $k^{(1)}$ by~(\ref{eq:phi-k-1-sol})--(\ref{eq:hk-deriv}) with recourse to 
formula~(\ref{eq:phik-exp-2}) with~(\ref{eq:check-Ck}) and~(\ref{eq:check-Cphi}). 

We start with~(\ref{eq:phi-k-1-sol}).
The requisite derivatives of $w$, $g$ and $h$ in the present case (where practically $\theta=0$) reduce to
\begin{equation}
w_\phi \sim 2l_{123} P^{-1}=O(P^{-1})~,\label{eq:wphi-rg2}
\end{equation}
\begin{equation}
w_k\sim -2P^{-1/3}\,(l_2^2+l_3^2)\check C^\phi_0=O(P^{-1/3})~,
\label{eq:wk-rg2}
\end{equation}
\begin{equation}
g_\phi\sim 4m_{123}P^{-1/3}\,\check C^\phi_0=O(P^{-1/3})~,
\label{eq:gphi-rg2}
\end{equation}
\begin{equation}
g_k\sim -2P^{1/3}(m_2 l_2+m_3 l_3) (\check C^\phi_0)^2=O(P^{1/3})~,
\label{eq:gk-rg2}
\end{equation}
\begin{equation}
h_\phi\sim {\textstyle\frac{1}{2}}n_{123}P^{-1/3} (\check C^k_0)^2=O(P^{-1/3})~,
\label{eq:hphi-rg2}
\end{equation}
\begin{equation}
h_k \sim n_{123}{\check C^k_0}\,{\check C^\phi_0}=O(1)~.
\label{eq:hk-rg2}
\end{equation}Note that $w^{(0)}$ is given by~(\ref{eq:w0-rg2}).

It follows that the determinants $\mathcal D$, $\mathcal D^\phi$ and $\mathcal D^k$ of
(\ref{eq:D-determ})--(\ref{eq:Dk-determ}) become
\begin{equation}
\mathcal D\sim -P^{-2/3} {\check C^\phi_0} l_{123}[n_{123} (\check C^k_0)^2+8 m_{123}\check C^\phi_0)=O(P^{-2/3})~,
\label{eq:D-rg2}
\end{equation}
\begin{equation}
\mathcal D^\phi \sim -4P^{-2/3} l_{123}^2 {\check C^\phi_0}\ (\check C^\phi_0\theta)_\theta\bigl|_{\theta=0}=O(P^{-2/3})~,
\label{eq:Dphi-rg2}
\end{equation}
\begin{equation}
\mathcal D^k \sim w_\phi k^{(0)} w^{(0)}\sim 4P^{-1}l_{123}^2 \check C^k_0 \check C^\phi_0=O(P^{-1})~.
\label{eq:Dk-rg2}
\end{equation}Since $\partial_\theta(\check C^\phi_0)$ is finite at $\theta=0$,~(\ref{eq:phi1-rg2})
and~(\ref{eq:k1-rg2}) ensue directly via~(\ref{eq:phi-k-1-sol}).
\vskip5pt

(iii)\ {\it Transition region,} $\theta=O(P^{1/3})$. Next, we study the limits of the
$\phi^{(1)}$ and $k^{(1)}$ found above when $\theta$ enters the transition region, $\theta\to O(P^{1/3})$.

First, we consider $\phi^{(1)}$ in the range $\theta >O(P^{1/3})$ and take $\theta\ll 1$. 
By~(\ref{eq:Cphi-def}) and~(\ref{eq:Q-def})--(\ref{eq:k0-asymp1}), we find $H^k=O(1)$,
$W^k=O(1/\theta)$, $W^\phi=O(1)$, and
\begin{equation}
w^{(0)}= \frac{f_+ + f_-}{\theta}+O(\theta)=O\biggl(\frac{1}{\theta}\biggr)\Rightarrow (w^{(0)}\theta)_\theta=
(f_+ +f_-)_{\theta}|_{\theta=0}+O(\theta)~,
\label{eq:w0-smalltheta}
\end{equation}
\begin{equation}
(v^{(0)}_\theta+v^{(0)}\tan\theta)k^{(0)}_\theta+(\cos\theta)^{-1}(w^{(0)}_\theta\tan\theta)_\theta=
2(f_+ +f_-)_\theta|_{\theta=0}+O(\theta)~.
\label{eq:num-smalltheta}
\end{equation}
Hence, assuming $(f_+ +f_-)_\theta\neq 0$ at $\theta=0$, we have
\begin{equation}
\phi^{(1)}=O(P/\theta^2)\cdot O((f_++f_-)_\theta)\qquad O(P^{1/3})<\theta\ll 1~,
\label{eq:phi1-ord-rg1}
\end{equation}
which becomes $O(P^{1/3}(f_+ +f_-)_\theta)$ as $\theta\to O(P^{1/3})$.
On the other hand, by~(\ref{eq:phi1-rg2}) we get $\phi^{(1)}=O(1)$
when $0\le \theta < O(P^{1/3})$. This behavior is not in agreement with~(\ref{eq:phi1-ord-rg1})
unless $(f_+ + f_-)_\theta = O(P^{-1/3})$, i.e. the
fluxes vary over angles $O(P^{1/3})$, $f_\pm=\breve f_\pm(P^{-1/3}\theta)$
for $\theta=O(P^{1/3})$. This behavior of $f_\pm$ is not compelling, since it is generally expected that the agreement
in orders of magnitude is spoiled by the $\theta$-differentiation.\looseness=-1

We next consider $k^{(1)}$. By~(\ref{eq:k1-app1}) we find $k^{(1)}=O((f_+ +f_-)_\theta)$ for
$O(P^{1/3})<\theta\ll 1$. On the other hand, by~(\ref{eq:k1-rg2}), $k^{(1)}=O(P^{-1/3})$ for $0\le\theta <O(P^{1/3})$.
The two orders of magnitude agree if $(f_+ + f_-)_\theta = O(P^{-1/3})$ as above.

\subsection{Condition on $\kappa$ and $P$}
\label{sssec:limitn}
Thus far, we have not provided any condition for the validity of our perturbation
analysis. Such a condition would impose a constraint on $\kappa$
and $P$. In principle, $\kappa$ is a dynamic variable. For appropriate initial data,
the step edges are assumed to evolve to the kinetic steady state with $\kappa=0$. Small deviations from this
state can be treated within our perturbation framework if
\begin{equation}
|\kappa\phi^{(1)}|\ll \phi^{(0)},\qquad |\kappa\,k^{(1)}|\ll k^{(0)}~.
\label{eq:phk-constr}
\end{equation}

By revisiting the formulas of sections~\ref{sec:steadyPdependence} and~\ref{sec:GT-stiff} for
$\phi^{(j)}$ and $k^{(j)}$, we can give an order-of-magnitude estimate of an upper bound for $\kappa$.
By comparison of the $O(P)$ correction term for 
$k^{(0)}$ in~(\ref{eq:k0-asymp1-corr}) to $k^{(1)}$ in~(\ref{eq:k1-app1}), where $\theta=O(1)$, we obtain 
\begin{equation}
|\kappa|< O(P)~.
\label{eq:restr-kappa}
\end{equation}

\subsection{Step stiffness}
\label{subsec:stiff}
Once $\phi^{(0)}$ and $\phi^{(1)}$ have been derived, the step stiffness follows.
We invoke the formulation of section~\ref{subsec:pertb}
on the basis of formula~(\ref{eq:stiff-intro}) by using the fluxes $f_\pm$
as input external parameters. In particular, we show the limiting behaviors~(\ref{eq:stiff-smalltheta})
and~(\ref{eq:stiff-rg2}) for small $\theta$. In correspondence to section~\ref{subsec:pertb},
we use two distinct regimes.
\vskip5pt

(i)\ $O(P^{1/3})<\theta <\pi/4$. By~(\ref{eq:lten}) and the analysis in section~\ref{subsec:pertb}, 
$\tbe$ is given by~(\ref{eq:stiff-intro})--(\ref{eq:k0-asymp1}). Specifically,
\begin{equation}
\frac{\tbe}{k_B T}\sim\frac{C^\phi_1}{C^\phi_0}~,\label{eq:stiff-rg1}
\end{equation}
which is an $O(1)$ quantity in $P$ when $\theta=O(1)$.
In order to compare this result to a recent equilibrium-based calculation
for the stiffness~\cite{stasevich06}, we take $O(P^{1/3})<\theta\ll 1$.
Then, by~(\ref{eq:Cphi-def}),
\begin{equation}
C_0^\phi\sim \frac{1}{\theta}\,\frac{\fp +\fm}{l_{123}}=O\biggl(\frac{1}{\theta}\biggr)~.
\label{C0-phi-lim}
\end{equation}
In addition, if $(f_+ +f_-)_\theta \neq 0$ as $\theta\to 0^+$, 
by~(\ref{eq:Cph1-def}) and~(\ref{eq:num-smalltheta}) we find
\begin{equation}
C_1^\phi\sim \frac{(\fp +\fm)_\theta|_{\theta=0}}{n_{123}\,\theta^2}=O\biggl(\frac{1}{\theta^2}\biggr)~.\label{eq:C1-phi-lim}
\end{equation}
Thus,~(\ref{eq:stiff-smalltheta}) follows from~(\ref{eq:stiff-rg1}).
By contrast, if $(f_+ + f_-)_\theta$ vanishes in the limit $\theta\to 0$ then, by~(\ref{eq:num-smalltheta}),
$\tbe/(k_B T)=O(1)$.
\vskip5pt

(ii)\ $0\le\theta < O(P^{1/3})$. In view of~(\ref{eq:stiff-intro}) with (\ref{eq:phik-exp-2})
and~(\ref{eq:phi1-rg2}), we readily obtain formula~(\ref{eq:stiff-rg2}) for $\tbe$.
\vskip5pt

(iii)\ $\theta\to O(P^{1/3})$. Formula~(\ref{eq:stiff-rg2}) 
is consistent with the $O(1/\theta)$ behavior of $\tbe$ for $O(P^{1/3})<\theta\ll 1$
provided that $(f_+ +f_-)_\theta =O(P^{-1/3})$. Indeed, from~(\ref{eq:stiff-rg1}) via~(\ref{eq:num-smalltheta})
we have $\tbe/(k_B T)= O((f_+ +f_-)_\theta/\theta)$, which properly reduces to~(\ref{eq:stiff-rg2}).
Again, this ``matching'' is not compelling since $\theta$-derivatives are involved.

\subsection{Alternative view}
\label{subsec:extension}
We consider $\theta=O(1)$ and focus briefly on the implications for the stiffness of treating
the adatom densities $\rho_\pm$ as input parameters. This approach is mathematically
equivalent to that of section~\ref{subsec:stiff}; only the physical definitions are altered
in recognition of $\rho_\pm$ as the driving parameters. This viewpoint was partly followed in
section~\ref{subsec:density} for straight step edges ($\kappa=0$).\looseness=-1

We show that the adatom fluxes have the form
\begin{equation}
f_\pm= \Da^{\pm}\rho_\pm\pm D_p^\pm (\rho_+ - \rho_-)-\be(\theta;\rho_+,\rho_-)\cdot\kappa~.
\label{eq:fpm-ss}
\end{equation}
The coefficients $D_p^\pm(\theta)$ and $\Da^\pm(\theta)$ are defined by~(\ref{eq:perm-Dpm}) 
and~(\ref{eq:Dpm-ES-dens}); and
\begin{equation}
\be=\frac{2 \,C^\phi_1}{1+(A_+ +A_-)\cos\theta}\,\cos\theta\qquad O(P^{1/3})<\theta<\pi/4~,
\label{eq:ss-def}
\end{equation}
where $C^\phi_1$ and $A_\pm$ are defined by~(\ref{eq:Cph1-def}) and~(\ref{eq:Apm-def}). 
Furthermore, the $f_\pm$-dependent $C^\phi_1$ is now evaluated
at $f_\pm=\Da^\pm\rho_\pm \pm D_p^\pm(\rho_+ -\rho_-)$; thus, $\be$ becomes $\rho$-dependent.  Notably,
\begin{equation}
\be = O(1/\theta)\qquad O(P^{1/3})<\theta\ll 1~.
\label{eq:ss-1overth}
\end{equation}
As noted in section \ref{stepStiffness}, these results do not have the usual form 
since $\be$ is not proportional to $\rho_*$.

We derive~(\ref{eq:fpm-ss})--(\ref{eq:ss-1overth}) directly from~(\ref{eq:fpm-kappa})
by treating the term $\kappa\,\phi^{(1)}$ as a perturbation.
For $\kappa\phi^{(1)}=0$ (section~\ref{subsec:density}), (\ref{eq:f-bcf}) for $f_\pm$ is recovered 
with $\rho_0=0$; see~(\ref{eq:rho-star-perm}). 
For $\kappa\neq 0$, (\ref{eq:fpm-kappa}) reads
\begin{equation}
f_\pm \sim (1+l_{j_\pm}\tan\theta) \Dt\rho_\pm\cos\theta-(A_+ f_+ + A_- f_-)\cos\theta-C^\phi_1\kappa\cos\theta~.
\label{eq:fpm-kappa-2}
\end{equation}By viewing $C^\phi_1$ as a given external parameter,
we solve the linear equations~(\ref{eq:fpm-kappa-2}) 
for $f_\pm$ and find~(\ref{eq:fpm-ss}) with~(\ref{eq:ss-def});
$\be$ follows as a function of $\rho_\pm$ by a single iteration.  

We now take $\theta\ll 1$. By~(\ref{eq:Apm-def}), $A_\pm=O(1/\theta)$ while by~(\ref{eq:phi1-ord-rg1})
we have $C^\phi_1=O(1/\theta^2)$ assuming $(\fp +\fm)_\theta=O(1)\neq 0$.
Thus, (\ref{eq:ss-def}) leads to~(\ref{eq:ss-1overth}).

Note that the standard Gibbs-Thomson formula~(\ref{eq:GT}) is not applicable here since
$\rho_*=0$ (and hence $\rho_0=0$). However, a linear-in-$\kappa$ term in $f_\pm$ is present, giving
rise to a ``generalized'' stiffness $\be$ that is not bound to a reference density $\rho_*$. 

\section{Conclusion}
\label{sec:conclusion}

The Gibbs-Thomson formula and stiffness of a step edge or island boundary were studied
systematically from an atomistic, kinetic perspective. Our starting point was a
kinetic model for out-of-equilibrium processes~\cite{caflischetal99,caflischli03}. The kinetic effects
considered here include diffusion of edge-atoms and convection of kinks along step edges, supplemented
with mean-field algebraic laws that relate mass fluxes to densities. 
Under the assumption that the model reaches a kinetic steady state with straight steps,
the step stiffness is determined
by perturbing this state for small edge curvature and P\'eclet number $P$ with
$|\kappa| < O(P)$, and applying the quasi-steady approximation . A noteworthy result is that for sufficiently
small $\theta$, $O(P^{1/3})<\theta\ll 1$, the step stiffness
behaves as $\tbe = O(1/\theta)$. This behavior is in qualitative agreement with independent calculations based
on equilibrium statistical mechanics~\cite{stasevich06,stasevichetal05}.

Our analysis offers the first derivation of the step stiffness, a near-equilibrium
concept, in the context of nonequilibrium kinetics. The results 
here are thus a 
step towards a better understanding of how evolution out of equilibrium
can be reconciled with concepts of equilibrium thermodynamics for crystal surfaces. Furthermore, this analysis 
provides a linkage of microscopic parameters, e.g. atomistic transition rates and coordination numbers,
to mesoscopic parameters of a BCF-type description. This simpler description is often a more
attractive alternative for numerical simulations of epitaxial growth.\looseness=-1

There are various aspects of the problem that were not addressed in our analysis. For instance,
it remains an open research direction to compare our predictions with results stemming from other kinetic
models~\cite{balykovvoigt05,balykovvoigt06,filimonovhervieu04}.
The existence of a kinetic steady state with straight edges, although expected
intuitively for a class of initial data, should be tested with 
numerical computations. Germane is the assumption of linear-in-$\kappa$
corrections in expansions for the associated densities.
Our perturbation analysis is limited
by the magnitudes of $\kappa$ and $P$; specifically, $|\kappa| < O(P)$. The
formal derivations need to be re-worked for $\kappa > O(P)$ as $P\to 0$.
The kinetic steady state here forms a basis solution for our perturbation theory,
and is different from an equilibrium state.
At equilibrium, detailed balance implies that the fluxes $f_+$, $f_-$
and each of the physical contributions (terms with different coordination
numbers) in~(\ref{eq:w-def})--(\ref{eq:h-def}) for $w$, $g$, and $h$
must vanish identically~\cite{caflischetal99}. An analysis based on this equilibrium approach and
comparisons with the present results are the subjects of work in progress. 
Generally, it also remains a challenge to compare in detail kinetic models such as ours 
with predictions put
forth by Kallunki and Krug with regard to the Einstein relation
for atom migration along a step edge~\cite{kallunkikrug03}. Our underlying step edge model
is based on a simple cubic lattice, and it does not include separate rates for kink or corner rounding.

Lastly, we mention two limitations inherent to our model.
The mean-field laws for the mass fluxes
are probably inadequate in physical situations where atom correlations are crucial.
The study of effects beyond mean field, a compelling but difficult task,
lies beyond our present scope. In the same vein, we expect that the effects of 
elasticity~\cite{connelletal06,kuktabat02,shenoy04}
will in principle modify the mesoscopic kinetic rates (attachment-detachment and permeability coefficients)
and the step stiffness. The inclusion of elastic effects in the kinetic model
and the study of their implications is a viable direction of near-future work.

\section*{Acknowledgments}We thank T.~L. Einstein, J. Krug, M.~S. Siegel, T.~J. Stasevich, 
A. Voigt, and P.~W. Voorhees for useful discussions. One of us (DM) is grateful for the hospitality
extended to him by the
Institute for Pure and Applied Mathematics (IPAM) at the 
University of California, Los Angeles, 
in the Fall 2005, when part of this work was completed.\looseness=-1

\appendix

\section{Step edge coordinates and basic relations}
\label{app:coods}
In this appendix we describe several coordinate systems for an island boundary,
thus supplementing the formulation of section~\ref{subsec:geom}.
Consider step boundaries that stem from perturbing a straight step edge parallel
to the $x$-axis; see Figure~\ref{fig:Fig1}.
Three associated coordinates and generic densities and longitudinal velocities (along the step edge) are defined as follows.
\begin{itemize}
\item {\em Fixed ($x$-) axis:} Variable $x$, velocity $w$, density $\xi$ ($\xi=\phi$ or $k$).
\item {\em Lagrangian:} Variable $\alpha$, velocity $W$, density $\Xi$.
\item {\em Arc length:} Variable $s$, velocity $\widetilde W$, density $\widetilde \Xi$.
\end{itemize}
The vector-valued normal velocity of the boundary is 
$v\, \norm$, where $\norm$ is defined in~(\ref{eq:unit-def}). 

We note the relations
\begin{eqnarray}
s_\alpha &=& \sqrt{x_\alpha^2 + y_\alpha^2}~,\nonumber\\
s_x &=& \sqrt{1+y_x^2} = 1/\cos\theta~,\label{eq:s-der}
\end{eqnarray}
\begin{equation}
x_s^2 + y_s^2 = 1\ \Rightarrow\
x_s x_{ss} + y_s y_{ss} = 0~.\label{eq:x-y-s}
\end{equation}
By use of the Lagrangian coordinate $\alpha$, we denote
\begin{equation}
\dt:=\partial_t|_\alpha,\qquad \partial_t:=\partial_t|_x~,\label{eq:part-t-den}
\end{equation}i.e., $d/dt$ is the time derivative with the spatial variable $\alpha$ held
fixed. Because the arc length, $s$, is only defined up to an arbitrary shift, 
we choose not to use a time derivative with $s$ held fixed; instead, we use $d/dt$ 
in conjunction with the $s$ derivatives.

Thus, the interface velocity $v\norm$ in the different coordinates is given by
\begin{eqnarray}
\partial_t |_\alpha (x,y)&=& v\,\norm = v (y_s, -x_s)~, \nonumber \\
\partial_t |_x (x,y)&=& v \norm + u_1\, \tang = (0,u_2)~,\label{eq:vel-coods}
\end{eqnarray}
in which $\tang$ is defined in~(\ref{eq:unit-def}) and
\begin{equation}
u_1= -vy_s/x_s=-v\tan\theta~,\qquad u_2= -v/x_s=-v/\cos\theta~.\label{eq:u12}
\end{equation}
The tangential derivatives and the time derivatives are related by
\begin{eqnarray} 
\partial _\alpha &=& s_\alpha \partial_s~,\qquad
\partial_x = s_x  \partial_s = (1/\cos\theta)\, \partial_s~,  \nonumber \\
\partial_t |_x&=& \partial_t |_\alpha +  (\partial_t |_x \alpha) \partial_\alpha 
=\partial_t |_\alpha +  (u_1/s_\alpha) \partial_\alpha~.\label{eq:tangt-der} 
\end{eqnarray}

We now use these relations to state transformation rules involving the $(\theta, t)$ variables;
see section~\ref{subsec:noneq-kin} for their applications. We assume that $\theta$ is a monotone function
of the coordinate $x$ and the arc length, $s$.
Useful derivatives in $x$ and $t$ are
\begin{equation}
\partial_x=s_x\,\partial_s=s_x\,\theta_s\,\partial_\theta=-\frac{\kappa}{\cos\theta}\,\partial_\theta,
\quad \partial_x^2=\frac{\kappa}{\cos\theta}\,\partial_\theta\,\frac{\kappa}{\cos\theta}\,\partial_\theta~,
\label{eq:partx}
\end{equation}
\begin{equation}
\partial_t|_x=\dt+(\partial_t|_x \alpha)\partial_\alpha=\dt+(\partial_t|_x\alpha)s_\alpha\,\partial_s=
\dt-v\tan\theta\,\theta_s\partial_\theta=\dt+v\kappa\tan\theta\,\partial_\theta~,
\label{eq:part-t}
\end{equation}
where $s_x=1/\cos\theta$, $\theta_s=-\kappa$, and 
\begin{equation}
\dt=\partial_t|_\alpha=\partial_t|_\theta+(\partial_t|_\alpha\theta)\,\partial_\theta~,
\label{eq:tot-t}
\end{equation}
\begin{equation}
\partial_t|_\alpha\theta=\dt\theta=\kappa\,v_\theta~.
\label{eq:part-ta}
\end{equation}By~(\ref{eq:tot-t}) and~(\ref{eq:part-ta}) we have
\begin{equation}
\dt=\partial_t|_\theta+\kappa\,v_\theta\,\partial_\theta~.
\label{eq:tot-t-fin}
\end{equation}Thus,~(\ref{eq:part-t}) becomes
\begin{equation}
\partial_t|_x=\partial_t|_\theta+\kappa(v\,\tan\theta+v_\theta)\,\partial_\theta~.
\label{eq:part-tx}
\end{equation}In particular,
\begin{equation}
\partial_t|_x\wp(\theta)=\kappa(v\,\tan\theta+v_\theta)\,\partial_\theta\wp~,
\label{eq:part-tx-th}
\end{equation}for any differentiable function $\wp(\theta)$ ($\partial_t|_\theta\wp\equiv 0$).

We close this appendix by deriving relations for the densities and velocities along a step edge in the different coordinates.
If $\xi$, $\Xi$ and $\widetilde\Xi$ denote line densities of the same atom species in $x$, $\alpha$ and $s$, 
we have $\xi\,dx=\Xi\,d\alpha = {\widetilde\Xi}\, ds$. Thus,\looseness=-1
\begin{eqnarray}
\Xi&=&x_\alpha \xi = x_s s_\alpha\,\xi = (\cos \theta) s_\alpha\, \xi~, \nonumber \\
{\widetilde \Xi}&=& x_s  \xi= (\cos \theta)\, \xi =\Xi/s_\alpha~.
\label{eq:dens-coods}
\end{eqnarray}
Next, we derive corresponding relations for the longitudinal velocities $w$, $W$ and $\widetilde W$.
If the position of a moving point is $X(t)$, $S(t)$ or ${\mathcal A} (t)$ in the $x$, $s$ and $\alpha$ coordinates,
respectively, then the velocities in these coordinates are related by
\begin{eqnarray}
w&=& X_t=v/y_s=v/\sin\theta~, \nonumber \\
W&=& {\mathcal A}_t = (w-vy_s)/x_\alpha~,\nonumber\\
{\widetilde W} &=& W s_\alpha~.
\label{eq:vels-coods}
\end{eqnarray} 

\section{Identities for step edge motion}
\label{app:id-edge}
In this appendix we state and prove three propositions pertaining to motion along a step
edge. Some of these results are used in relation to section~\ref{subsec:noneq-kin} and in appendix~\ref{app:eq-mot}
in order to derive alternative equations of motion for edge-atom and kink densities.

\begin{proposition}In the $(x,t)$ variables, the step edge velocity $v$ satisfies
\begin{equation}
\partial_t(\tan\theta)+\partial_x(v/\cos\theta)=0~.\label{eq:cons-v}
\end{equation}
\label{prop:B1}
\end{proposition}

\begin{proof}
We proceed by direct evaluation of the derivatives appearing in~(\ref{eq:cons-v}).
First, we calculate the time derivative in terms of $s$ derivatives via the relation
\begin{equation}
\partial_t|_x \tan\theta = \biggl(\partial_t |_\alpha - \frac{vy_s}{s_\alpha x_s}\,\partial_\alpha\biggr) \tan\theta~.
\label{eq:prop1-1}
\end{equation}
By~(\ref{eq:x-y-s}) of appendix~\ref{app:coods}, we evaluate separately each term in the right-hand side:
\begin{eqnarray}
\partial_t |_\alpha \tan\theta &=&\partial_t (y_\alpha/x_\alpha) 
=(y_{\alpha t} x_{\alpha} - x_{\alpha t} y_{\alpha} )/x_{\alpha}^2
\nonumber  \\
&=& [(-vx_\alpha/s_\alpha)_{ \alpha} x_{\alpha} - (vy_\alpha/s_\alpha)_{ \alpha} y_{\alpha} ]/x_{\alpha}^2
=-[(vx_s)_{s} x_{s} + (vy_s)_{ s} y_{s} ]/x_{s}^2
\nonumber  \\
&=&-[v_{s} (x_{s}^2+y_{s}^2) + v (x_{ss}x_s + y_{ss}y_s )]/x_{s}^2
= -v_s/x_s^2
= -v_s/\cos^2 \theta~,\label{eq:prop1-2}
\end{eqnarray}
\begin{eqnarray}
\partial_\alpha \tan\theta
&=& s_\alpha \partial_s (y_s/x_s)
=s_\alpha ( x_s y_{ss} - y_s x_{ss})/x_s^2
\nonumber  \\
&=& s_\alpha [ x_s (-x_s x_{ss}/y_s) - y_s x_{ss}]/x_s^2
= -\frac{s_\alpha  x_{ss}}{y_s x_s^2}=
-\frac{s_\alpha  x_{ss}}{\sin \theta \cos^2 \theta}~.
\label{eq:prop1-3}
\end{eqnarray}

Second, we address the spatial derivative in~(\ref{eq:cons-v}):
\begin{equation}
\partial_x (v/\cos\theta)
= \frac{1}{\cos\theta} \partial_s \frac{v}{\cos\theta}
= \frac{v_s}{\cos^2\theta}+\frac{v}{\cos\theta} \partial_s \frac{1}{x_s}
= \frac{v_s}{\cos^2\theta}-v\frac{x_{ss}}{\cos^3\theta}~.
\label{eq:prop1-4} 
\end{equation}
Equations~(\ref{eq:prop1-1})--(\ref{eq:prop1-4}) combined yield~(\ref{eq:cons-v}),
which completes the proof.
\label{prf:prop-B1}\hfill
\end{proof}

\begin{proposition}
If the line (step-edge) density $\xi(x,t)$ satisfies
\begin{equation}
\partial_t \xi + \partial_x (w\xi) =b~,\label{eq:cons-phi-x}
\end{equation}
in the $(x,t)$ coordinates, then the following relations hold  
in $(\alpha,t)$ and $(s,t)$:
\begin{equation}
(d/dt)\Xi + \partial_\alpha (W\Xi) =B:=b\, x_\alpha~,
\label{eq:prop2-res1}
\end{equation}
\begin{equation}
(d/dt) {\widetilde \Xi} + \partial_s ({\widetilde W}{\widetilde \Xi})+v \kappa {\widetilde \Xi} ={\widetilde B}:=B/ s_\alpha~,
\label{eq:prop2-res2}
\end{equation}where $\partial_t=\partial_t|_x$ and $d/dt=\partial_t|_\alpha$.
\label{prop:B2}
\end{proposition}

\begin{proof}
We proceed by direct evaluation of the left-hand side of~(\ref{eq:cons-phi-x}) in the $(\alpha,t)$ and $(s,t)$ coordinates.
First, we prove~(\ref{eq:prop2-res1}).
By~(\ref{eq:dens-coods}) we have 
$\xi=\Xi/x_\alpha$ and $w\xi=(Wx_\alpha+vy_s)\Xi/x_\alpha=W\Xi+(vy_\alpha/x_\alpha s_\alpha)\Xi$. 
Thus, the time derivative in~(\ref{eq:cons-phi-x}) becomes
\begin{eqnarray}
\partial_t|_x \xi
&=& \biggl(\partial_t |_\alpha - \frac{v\,y_s}{s_\alpha x_s}\,\partial_\alpha\biggr)\,\frac{\Xi}{x_\alpha}
\nonumber  \\
&=& \frac{\Xi_t}{x_\alpha} - \frac{vy_\alpha}{s_\alpha x_\alpha^2}\,\Xi_\alpha 
+ \Xi\biggl[\partial_t\biggl(\frac{1}{x_\alpha}\biggr)-\frac{vy_s}{s_\alpha x_s}\,\partial_\alpha \biggl(\frac{1}{x_\alpha}\biggr)\biggr]~.
\label{eq:prop2-1}
\end{eqnarray}
Similarly, the spatial derivative in~(\ref{eq:cons-phi-x}) reads 
\begin{eqnarray}
\partial_x (w\xi)
&=& \frac{1}{s_\alpha \cos\theta}\,\partial_\alpha\biggl(W\Xi+\frac{vy_\alpha}{x_\alpha s_\alpha}\Xi \biggr)
\nonumber  \\
&=& \frac{1}{x_\alpha}\,\partial_\alpha(W\Xi)+\frac{vy_\alpha}{x_\alpha^2 s_\alpha}\,\Xi_\alpha 
+\frac{\Xi}{x_\alpha}\,\partial_\alpha\biggl(\frac{vy_\alpha}{x_\alpha s_\alpha}\biggr)~.
\label{eq:prop2-2}
\end{eqnarray}

The combination of~(\ref{eq:prop2-1}) and~(\ref{eq:prop2-2}) yields
\begin{eqnarray}
\lefteqn{\partial_t \xi + \partial_x (w\xi)-(1/x_\alpha)[\partial_t \Xi + \partial_\alpha (W\Xi)]}\nonumber\\
&=& \Xi\,\biggl[\partial_t\biggl(\frac{1}{x_\alpha}\biggr)-\frac{vy_s}{s_\alpha x_s}\,\partial_\alpha\biggl(\frac{1}{x_\alpha}\biggr) 
+\frac{1}{x_\alpha}\,\partial_\alpha\biggl(\frac{vy_\alpha}{x_\alpha s_\alpha}\biggr)\biggr]\nonumber\\
&=:& {\mathcal C}_1\,v +{\mathcal C}_2\, v_\alpha~.
\label{eq:prop2-3}
\end{eqnarray}The last expression is justified by the identity 
\begin{eqnarray}
\partial_t(1/x_\alpha)&=& 
-(x_t)_\alpha /x_\alpha^2 = -(vy_s)_s s_\alpha/x_\alpha^2 = (-v_sy_ss_\alpha - v y_{ss} s_\alpha)/x_\alpha^2
\nonumber  \\
&=& -v_\alpha\,\frac{y_\alpha}{s_\alpha x_\alpha^2} - v\,\frac{y_{ss} s_\alpha}{x_\alpha^2}~.
\label{eq:prop2-4}
\end{eqnarray}
Next, we show that the coefficients ${\mathcal C}_j$ in~(\ref{eq:prop2-3}) vanish identically.
To this end, we convert the related $s$ derivatives to $\alpha$ derivatives via the identities
\begin{equation}
y_{ss}=(y_\alpha/s_\alpha)_\alpha/s_\alpha = y_{\alpha\alpha}/s_\alpha^2 - y_\alpha s_{\alpha\alpha}/s_\alpha^3~,
\label{eq:prop2-5}
\end{equation}
\begin{equation}
\partial_\alpha \biggl(\frac{v y_\alpha }{ s_\alpha x_\alpha}\biggr)
=v_\alpha \frac{ y_\alpha}{s_\alpha x_\alpha} 
+v\frac{y_{\alpha\alpha}x_\alpha s_\alpha - y_\alpha(x_\alpha s_{\alpha\alpha}+s_\alpha x_{\alpha\alpha})  }
{x_\alpha^2 s_\alpha^2}~.
\label{eq:prop2-6}
\end{equation}
It follows by~(\ref{eq:prop2-3}) that ${\mathcal C}_j$ are 
\begin{equation}
{\mathcal C}_2=-\frac{y_\alpha}{s_\alpha x_\alpha^2} + \frac{1}{x_\alpha} \frac{y_\alpha}{s_\alpha x_\alpha}=0~,
\label{eq:prop2-7}
\end{equation}
\begin{equation}
{\mathcal C}_1= -\frac{y_{ss} s_\alpha}{x_\alpha^2} - \frac{y_s}{s_\alpha x_s} \frac{-x_{\alpha\alpha}}{x_\alpha^2}+ \frac{1}{x_\alpha} 
\frac{y_{\alpha\alpha}x_\alpha s_\alpha - y_\alpha(x_\alpha s_{\alpha\alpha}+s_\alpha 
x_{\alpha\alpha})}{x_\alpha^2 s_\alpha^2}= 0~,
\label{eq:prop2-8}
\end{equation}
which in view of~(\ref{eq:cons-phi-x}) and~(\ref{eq:prop2-3}) yield~(\ref{eq:prop2-res1}).

To derive~(\ref{eq:prop2-res2}), we first note that $W\Xi={\widetilde W}\,{\widetilde \Xi}$. 
For arbitrary $\alpha_1$, $\alpha_2$ we consider the integral
\begin{eqnarray}
\int_{\alpha_1}^{\alpha_2}( \partial_t |_\alpha {\widetilde \Xi})\,ds 
&=&\int_{\alpha_1}^{\alpha_2}( \partial_t |_\alpha {\widetilde \Xi})s_\alpha\, d\alpha
=\int_{\alpha_1}^{\alpha_2}[\partial_t |_\alpha ({\widetilde\Xi}s_\alpha) - {\widetilde\Xi}( \partial_t |_\alpha 
s_\alpha)]\, d\alpha\nonumber \\
&=&\int_{\alpha_1}^{\alpha_2}( \partial_t |_\alpha \Xi - {\widetilde \Xi}\,v \kappa s_\alpha)\, d\alpha
=\int_{\alpha_1}^{\alpha_2} [-\partial_\alpha (W\Xi)+B - {\widetilde \Xi}v \kappa s_\alpha]\, d\alpha
\nonumber \\
&=&\int_{\alpha_1}^{\alpha_2}[ -\partial_s ({\widetilde W}{\widetilde \Xi})+{\widetilde B} - {\widetilde \Xi}v \kappa] 
s_\alpha\, d\alpha
=\int_{\alpha_1}^{\alpha_2}[ -\partial_s ({\widetilde W}{\widetilde \Xi})+{\widetilde B} - {\widetilde \Xi}v \kappa]\, ds~,
\label{eq:prop2-9}
\end{eqnarray}where we used $\partial_t|_\alpha s_\alpha=v\kappa s_\alpha$ and~(\ref{eq:prop2-res1}).
Equation~(\ref{eq:prop2-res2}) follows directly, thus concluding the proof.
\label{prf:prop-B2}\hfill
\end{proof}

\begin{proposition}
If the line density $\xi(x,t)$ satisfies
\begin{equation}
\partial_t \xi- \partial_{x} (d\partial_{x}\xi) =b~,
\label{eq:cons-xi-d}
\end{equation}
then the following relations hold: 
\begin{equation}
(d/dt) \Xi - \partial_{x} (D\partial_{x}\xi) + \partial_\alpha (U\Xi) =B:=b x_\alpha~,
\label{eq:prop3-res1}
\end{equation}
\begin{equation}
(d/dt) {\widetilde \Xi} - \partial_s ({\widetilde D}\partial_s{\widetilde \Xi})
+ \partial_s ({\widetilde U}{\widetilde \Xi})+v \kappa\, {\widetilde \Xi} ={\widetilde B}:=b/ s_\alpha~,
\label{eq:prop3-res2}
\end{equation}
where $\partial_t=\partial_t|_x$, $d/dt=\partial_t|_\alpha$, and
\begin{eqnarray}
D&:=&\frac{d}{x_\alpha^2}~,\qquad 
U:=d\,\biggl(\frac{x_{\alpha\alpha}}{x_\alpha^3} - \frac{vy_s}{x_\alpha}\biggr)~,\nonumber\\
{\widetilde D}&:=&\frac{d}{x_s^{2}}~,\qquad
{\widetilde U}:= d\,s_\alpha \biggl(-\frac{s_{\alpha\alpha}}{s_\alpha x_\alpha^2} +\frac{x_{\alpha\alpha}}{x_\alpha^3} 
- \frac{vy_s}{x_\alpha}\biggr)~.
\label{eq:prop3-def}
\end{eqnarray}
\label{prop:B3}
\end{proposition}

\begin{proof}
Equations~(\ref{eq:prop3-res1}) and~(\ref{eq:prop3-res2}) follow directly from Proposition~\ref{prop:B2} by setting $w=-\xi_x/\xi$. 
Indeed, with this substitution we have
\begin{eqnarray}
\Xi W&=&x_\alpha \xi\,\frac{w-vy_s}{x_\alpha}= \xi w-\xi vy_s = -\partial_x f - fvy_s
\nonumber \\
&=& -\partial_x\biggl(\frac{\Xi}{x_\alpha}\biggr) - \frac{\Xi}{x_\alpha}\,vy_s
= -\frac{1}{x_\alpha}\,\partial_\alpha\biggl(\frac{\Xi}{x_\alpha}\biggr) - \frac{\Xi}{x_\alpha}\,vy_s
\nonumber \\
&=& -\frac{1}{x_\alpha^2}\,\partial_\alpha \Xi + \Xi\,
\biggl(\frac{x_{\alpha\alpha}}{x_\alpha^3}-\frac{vy_s}{x_\alpha} \biggr)
\label{eq:prop3-1} \\
&=& -\frac{1}{x_\alpha^2}\,s_\alpha\partial_s (s_\alpha {\widetilde \Xi}) + s_\alpha\, {\widetilde \Xi}\,
\biggl(\frac{x_{\alpha\alpha}}{x_\alpha^3}-\frac{vy_s}{x_\alpha} \biggr)
\nonumber \\
&=&-\frac{s_\alpha^2}{x_\alpha^2}\,\partial_s  {\widetilde \Xi} 
+  {\widetilde \Xi}\,s_\alpha
\biggl(-\frac{s_{\alpha\alpha}}{s_\alpha x_\alpha^3} +\frac{x_{\alpha\alpha}}{x_\alpha^3}-\frac{vy_s}{x_\alpha} \biggr)~,
\label{eq:prop3-2}
\end{eqnarray}
where $\partial_s s_\alpha = s_{\alpha\alpha} / s_\alpha$ was used in the last line. 
Equation~(\ref{eq:prop3-res1}) 
comes from Proposition~\ref{prop:B2} in view of~(\ref{eq:prop3-1}).
Equation~(\ref{eq:prop3-res2}) stems from Proposition~\ref{prop:B2} via~(\ref{eq:prop3-2}),
which completes the proof.
\label{prf:prop-B3}\hfill
\end{proof}

\section{Edge-atom and kink motion in $s$ coordinate}
\label{app:eq-mot}
Next, we apply the results of appendix~\ref{app:id-edge} to transform evolution
laws~(\ref{eq:phi-pde}) and~(\ref{eq:k-pde}) for $\phi$ and $k$ to the $(s,t)$ coordinates.

By use of Proposition~\ref{prop:B2}, the convection equation~(\ref{eq:k-pde}) becomes
\begin{equation}
(d/dt)\widetilde K+\partial_s[\widetilde W(\widetilde K_r-\widetilde K_l)]+v\kappa\widetilde K=2\widetilde F_k~,
\label{eq:k-pde-s}
\end{equation}
where
\begin{equation}
\widetilde F_k:=(g-h)x_s=(g-h)\cos\theta~,\quad \widetilde W:=\frac{w-vy_s}{x_s}=\frac{w-v\sin\theta}{\cos\theta}~.
\label{eq:Fk-def}
\end{equation}
In the above, $\widetilde K=k\cos\theta$, $\widetilde K_r=k_r\cos\theta$, 
$\widetilde K_l=k_l\cos\theta$ and $\widetilde W$ are
defined along the edge arc length ($s$) according to the notation of appendix~\ref{app:coods}.

Proposition~\ref{prop:B3} of appendix~\ref{app:id-edge} converts the diffusion
equation~(\ref{eq:phi-pde}) for $\phi$ to
\begin{equation}
(d/dt)\widetilde\Phi-\partial_s(\widetilde D_E\partial_s\widetilde\Phi)+\partial_s(\widetilde U\widetilde\Phi)
+v\kappa\widetilde\Phi=\widetilde F_\phi~,
\label{eq:phi-pde-s}
\end{equation}
where
\begin{equation}
\widetilde F_\phi:=\fp+\fm-f_0\cos\theta~,\quad \widetilde D_E:=\frac{\De}{x_s^2}=\frac{\De}{\cos^2\theta}~,
\label{eq:Fphi-def}
\end{equation}
\begin{equation}
\widetilde U:=\De\,s_\alpha\biggl(-\frac{s_{\alpha\alpha}}{s_\alpha x_\alpha^2}+
\frac{x_{\alpha\alpha}}{x_\alpha^3}-\frac{vy_s}{x_\alpha}\biggr)~,
\label{eq:U-def}
\end{equation}
and $\alpha$ is the Lagrangian step coordinate; see appendix~\ref{app:coods}.
Here, $\widetilde \Phi$ is the edge-atom density defined along the edge arc length.
Note that the transformed equation~(\ref{eq:phi-pde-s}) contains a drift term, which
is absent in~(\ref{eq:phi-pde}) if $x$ is simply replaced by $s$.

\section{Step edge velocity}
\label{app:vel}
In this appendix we derive~(\ref{eq:v-def}) in the form of a proposition; cf. equation~(2.12) in~\cite{caflischli03}.

\begin{proposition}
The net flux $f_0$ of terrace and edge-atoms to kinks is 
\begin{equation}
f_0=\frac{v}{\cos\theta}(1+\kappa\phi\cos\theta)~.
\label{eq:f0-prop}
\end{equation}
\label{prop:D1}
\end{proposition}

\begin{proof}
We apply mass conservation, revisiting the derivation in~\cite{caflischetal99}.
The starting point is the change of the total number of adatoms on a terrace,
which is balanced by: (i) the step edge motion, (ii) the change 
of the number of edge-atoms, and (iii) the flux rate of deposited atoms. Hence,
\begin{equation}
-\dt\int\rho\,dA=\int_{\Gamma}v\,ds+\dt\int_\Gamma \widetilde\Phi\,ds-FA~,
\label{eq:drho}
\end{equation}
where $A$ is the area of a single terrace and $\Gamma$ is the step boundary.

Next, we find alternative expressions for the terms $\dt\int \rho\,dA$ and $\dt\int\widetilde\Phi\,ds$. 
First, integration of the diffusion equation~(\ref{eq:rho-pde}) for $\rho$ yields
\begin{equation}
\dt\int \rho\,dA=-\int_\Gamma (\fp+\fm)\,ds+FA~.
\label{eq:drho-2}
\end{equation}
Second, direct differentiation of $\int \widetilde\Phi\,ds$ with
respect to time gives
\begin{equation}
\dt\int_\Gamma \widetilde\Phi\,ds=\int_\Gamma (\partial_t|_s\widetilde\Phi+\kappa v\widetilde\Phi)\,ds~,
\label{eq:dphi}
\end{equation}
by using
\begin{equation}
\dt ds=\partial_t|_\alpha s_\alpha\,d\alpha=(\partial_t|_\alpha s_\alpha)\,d\alpha=(\kappa v s_\alpha)d\alpha=\kappa v\, ds~.
\label{eq:dsa}
\end{equation}
By combination of~(\ref{eq:drho})--(\ref{eq:dphi}) we obtain
\begin{eqnarray}
-\int_\Gamma v\,ds&=&-\int_\Gamma (\fp+\fm)\,ds+
\int_\Gamma \biggl[\dt\widetilde\Phi-(\partial_t|_\alpha s)\partial_s\widetilde\Phi+\kappa v\widetilde\Phi\biggr]
\nonumber\\
&=& -\int_\Gamma (\fp+\fm)\,ds+\int_\Gamma [\widetilde F_\phi-(\partial_t|_\alpha s)\partial_s\widetilde\Phi]\,ds~,\label{eq:dv}
\end{eqnarray}
where we invoked the evolution equation~(\ref{eq:phi-pde-s}) for $\widetilde\Phi$ in ($s,t$) coordinates and 
definition~(\ref{eq:Fphi-def}) from appendix~\ref{app:eq-mot}. 
Thus, via integration by parts,~(\ref{eq:dv}) becomes
\begin{eqnarray}
-\int_\Gamma v\,ds&=&-\int_\Gamma f_0\,x_s\,ds+\int_\Gamma \widetilde\Phi\,\partial_s(\partial_t|_\alpha s)\,ds\nonumber\\
&=&-\int_\Gamma f_0\,x_s\,ds+\int_\Gamma \widetilde\Phi\ s_\alpha^{-1}\partial_\alpha(\partial_t|_\alpha s)\ ds\nonumber\\
&=& -\int_\Gamma f_0\,x_s\ ds+\int_\Gamma \widetilde\Phi\kappa\,v\ ds~.
\label{eq:dv-2}
\end{eqnarray}
Hence, we have
\begin{equation}
-v=-f_0\,x_s+\widetilde\Phi\kappa\,v\qquad (\widetilde\Phi=\phi\, x_s=\phi\,\cos\theta)~,
\label{eq:v-k-phi}
\end{equation}which is identified with~(\ref{eq:v-def}) and, thus, concludes the proof.
\label{prf:prop-D1}\hfill
\end{proof}

\section{First-order perturbation theory}
\label{app:lin-pert}
In this appendix we describe in the form of a proposition the basic linear perturbation for the 
equations of motion along a step edge. This theory is used in section~\ref{sec:GT-stiff}.

\begin{proposition}
Let $\phi$ and $k$ be functions of $(\theta, t)$ that satisfy 
\begin{equation}
M_j(\phi, k, \kappa \phi_\theta, \kappa k_\theta)=0\qquad j=1,\,2~,
\label{eq:Mj}
\end{equation}where $M_j(\phi, k, \eta, \zeta)$ are differentiable. If~(\ref{eq:phk-app}) holds,
where $\phi^{(0)}$ and $k^{(0)}$ solve
\begin{equation}
M_j(\phi^{(0)}, k^{(0)}, 0, 0) = 0~,\label{eq:zeroth}
\end{equation}
then $\phi^{(1)}$ and $k^{(1)}$ are
\begin{equation}
\phi^{(1)}=\frac{\mathcal D^\phi}{\mathcal D}~,\qquad k^{(1)}=\frac{\mathcal D^k}{\mathcal D}~,
\label{eq:phk-1-app}
\end{equation}
where
\begin{equation}
\mathcal D=\left|\begin{array}{lr}\partial_\phi M_{1}& \partial_k M_{1}\\
                                  \partial_\phi M_{2}& \partial_k M_{2}\end{array}\right|~,
\label{eq:det-def}
\end{equation}
\begin{equation}
\mathcal D^\phi=\left|\begin{array}{lr}-(\partial_\zeta M_{1})\,\partial_\theta\phi_{0}-(\partial_\eta M_{1})\,\partial_\theta k_{0}& \ \partial_k M_{1}\\
                                  -(\partial_\zeta M_{2})\,\partial_\theta\phi_{0}-(\partial_\eta M_{2})\,\partial_\theta k_{0}& \ \partial_k M_{2}\end{array}\right|~,
\label{eq:detphi-def}
\end{equation}
\begin{equation}
\mathcal D^k=\left|\begin{array}{lr}\partial_\phi M_1& \ -(\partial_\zeta M_{1})\,\partial_\theta\phi_{0}-(\partial_\eta M_{1})\,\partial_\theta k_{0}\\
                                  \partial_\phi M_2& \ -(\partial_\zeta M_{2})\,\partial_\theta\phi_{0}-(\partial_\eta M_{2})\,\partial_\theta k_{0}\end{array}\right|~,
\label{eq:detk-def}
\end{equation}         
and the derivatives of $M_j(\phi, k, \eta, \zeta)$ are evaluated at $(\phi^{(0)}, k^{(0)}, 0, 0)$. 
\label{prop:E1}
\end{proposition}

\begin{proof}
Equations~(\ref{eq:phk-1-app})--(\ref{eq:detk-def}) follow directly from the 
Taylor expansion of formula~(\ref{eq:Mj}) at $(\phi^{(0)}, k^{(0)}, 0, 0)$,
\begin{equation}
0 = (\kappa\phi^{(1)})\,\partial_\phi M_j+(\kappa k^{(1)})\,\partial_k M_j+
(\kappa\partial_\theta\phi^{(0)})\partial_\zeta M_j+(\kappa \partial_\theta k^{(0)})\partial_\eta M_j~,
\label{eq:taylorMj}
\end{equation}
where use was made of~(\ref{eq:zeroth}). The $2\times 2$ linear system for $(\phi^{(1)}, k^{(1)})$
leads to~(\ref{eq:phk-1-app}).
\label{prf:prop-E1}\hfill
\end{proof}

\end{document}